\documentclass[accepted]{uai2026} 

\usepackage[american]{babel}

\usepackage{natbib} 
\usepackage{mathtools} 
\usepackage{booktabs} 
\usepackage{tikz} 

\usepackage[T1]{fontenc}
\usepackage[utf8]{inputenc}
\usepackage{mathrsfs}
\usepackage{amssymb}
\usepackage{amsthm}
\usepackage{siunitx} 
\usepackage{fancyhdr}
\usepackage[normalem]{ulem}
\usepackage{algorithm} 
\usepackage{algpseudocode} 
\usepackage{cancel} 
\usepackage{listings} 
\usepackage{graphicx}

\newenvironment{thma}[1]{\par\noindent{\bf Theorem #1.\ }\em}{\em}
\newenvironment{lma}[1]{\par\noindent{\bf Lemma #1.\ }\em}{\em}

\newtheorem{example}{Example} 
\newtheorem{theorem}{Theorem}
\newtheorem{lemma}[theorem]{Lemma}

\DeclareMathOperator{\Perp}{\perp\!\!\!\perp}

\DeclareMathOperator{\pa}{\operatorname{pa}}
\DeclareMathOperator{\ch}{\operatorname{ch}}

\DeclareMathOperator{\de}{\operatorname{de}}
\DeclareMathOperator{\nd}{\operatorname{nd}}
\DeclareMathOperator{\nei}{\operatorname{ne}}
\DeclareMathOperator{\bd}{\operatorname{bd}}

\DeclareMathOperator{\an}{\operatorname{an}}
\DeclareMathOperator{\sib}{\operatorname{sib}}
\DeclareMathOperator{\dis}{\operatorname{dis}}
\DeclareMathOperator{\mb}{\operatorname{mb}}

\DeclareMathOperator*{\EE}{\mathbb{E}}

\title{
A Characterization of the Orthocomplement of the Tangent Space of Semiparametric Markov Models
}

\author[1]{
    \href{mailto:<tphung1@jhu.edu>?Subject=Your UAI 2026 paper}
    {Trung Phung}
}
\author[1]{
    \href{mailto:<ilyas@cs.jhu.edu>?Subject=Your UAI 2026 paper}
    {Ilya Shpitser{}}
}
\affil[1]{%
    Computer Science Department\\
    Johns Hopkins University\\
    Baltimore, Maryland, USA
}
\begin{document}
\maketitle

\begin{abstract}
Graphical models are ubiquitous in social and empirical science as they are intuitive and easy to use. These models belong to the broader class of Markov models, defined using solely conditional independence (CI) restrictions.

In order to estimate finite-dimensional target parameters in such models efficiently, semi-parametric theory provides a principled framework for constructing regular and asymptotically linear estimators via influence functions (IFs). These estimators are asymptotically normal and root-$n$ consistent. Characterizing the class of all influence functions for a target parameter is crucial for statistically efficient inference in these models.

For models that are Markov relative to directed acyclic graphs (DAGs), the orthogonal complement of the tangent space is known, implying that for any target the class of all influence functions can be derived once an influence function is obtained. On the other hand, for Markov models not equivalent to a DAG model -- such as ordinary Markov models associated with undirected graphs, chain graphs, or acyclic directed mixed graphs -- the orthogonal complement has not been characterized, impeding semi-parametric inference in these models.

We derive closed form expressions for the orthogonal complement of the tangent space for general Markov models and illustrate our results by characterizing the class of influence functions for the conditional mean parameter in several graphical models.

\end{abstract}

\section{Introduction}\label{sec:intro}

A common task in statistical and causal inference is estimating a finite dimensional target parameter from data.  In parametric models where elements are indexed by finite dimensional nuisance parameters, efficient estimation is achieved using the theory of maximum likelihood estimation \citep{vanderVaart.2000.AsymptoticStatistics}.  In nonparametric or semiparametric models where the parameterization is infinite dimensional, regular and asymptotically linear (RAL) estimators are desirable, since they have nice properties such as asymptotic normality and root-$n$ rate of convergence. RAL estimators are typically obtained by 
deriving influence function of the target parameter
\citep{Newey.1990.SemiparametricEfficiency,Bickel.Klaassen.ea.1993.EfficientAdaptive,VanDerVaart.2002.SemiparametricStatistics,Tsiatis.2006.SemiparametricTheory,Kennedy.2024.SemiparametricDoubly}.

The class of influence functions of a parameter of interest in a given semiparametric model are elements of the linear variety constructed using one influence function and the orthogonal complement of the tangent space of the model (\citet{Tsiatis.2006.SemiparametricTheory}, Theorem 4.3). 
Characterizing the class of all possible RAL estimators in a given semiparametric model thus entails characterizing the orthogonal complement of the tangent space of the model.

Statistical independence is an important tool for encoding \emph{irrelevance,} and is thus crucial for obtaining interpretable models for empirical phenomena where relationships among variables are local.  Thus,
in this paper, we consider
Markov models, or conditional independence models, which are statistical models defined using solely marginal and conditional independences, with a focus on the subclass of Markov models associated with graphs, referred to as \textit{graphical Markov models} \citep{Pearl.1988.ProbabilisticReasoning,Lauritzen.1996.GraphicalModels,Maathuis.Drton.ea.2018.HandbookGraphical}.

Graphical models have been associated with directed acyclic graphs (DAGs), undirected graphs (UGs), as well as mixed graphs such as chain graphs (CGs) \citep{Frydenberg.1990.ChainGraph} and acyclic directed mixed graphs (ADMGs) \citep{Pearl.1988.ProbabilisticReasoning,Lauritzen.1996.GraphicalModels,Richardson.2003.MarkovProperties,Maathuis.Drton.ea.2018.HandbookGraphical}. Models associated with DAGs have seen wide use due to their interpretability, and connections to causal models \citep{Pearl.2000.CausalityModels,Spirtes.Glymour.ea.2001.CausationPrediction,Richardson.Robins.2013.SingleWorld}. However, other types of graphical models have also seen uses for associational analysis of gene regulatory networks \citep{segal2003mni}, protein signaling networks \citep{sachs05causal},
network data \citep{TchetgenTchetgen.Fulcher.ea.2021.AutoGComputationCausal,Ogburn.Shpitser.ea.2020.CausalInference}, and spatial data \citep{lopes08spatial}, as well as modeling systems at equilibrium \citep{Lauritzen.Richardson.2002.ChainGraph}. In addition, ADMG models have been connected to causal models with hidden variables, as well as theory of identification in such models \citep{Richardson.2003.MarkovProperties,Richardson.Evans.ea.2023.NestedMarkov,Shpitser.Richardson.ea.2022.MultivariateCounterfactual,Shpitser.Pearl.2006.IdentificationJointa}.

Explicit likelihoods in terms of sets of variationally independent components have been derived for DAG models, leading to a natural characterization of their model tangent spaces in terms of orthogonal subspaces \citep{Tsiatis.2006.SemiparametricTheory,Rotnitzky.Smucler.2020.EfficientAdjustment,Bhattacharya.Nabi.ea.2022.SemiparametricInference}. While such likelihoods are also known for UG and CG models \citep{Shpitser.2023.LauritzenChenLikelihood}, they feature overlapping components, meaning their tangent space cannot easily be expressed in terms of orthogonal subspaces.  Furthermore, no general likelihoods for ADMG models are known.  Since these Markov models associated with UGs, CGs, and ADMGs are not, in general, observationally equivalent to models associated with DAGs, this greatly complicates the derivation of the tangent space and hence its orthogonal complement, necessitating the development of new techniques.

In this paper, we derive a characterization of the orthogonal complement of the tangent space for semiparametric Markov models defined by marginal and conditional independences, including graphical models associated with undirected graphs, chain graphs, and acyclic directed mixed graphs. In addition, we illustrate how our characterization may be used to derive
all influence functions for target parameters arising in statistical and causal inference 
applications, with the conditional mean parameter as a worked example.

\section{Preliminaries}


Prior to discussing our contribution, we first review the necessary preliminaries: conditional independence (Markov) models and the important subsclass of graphical Markov models, as well as the theory of statistical inference in semi-parametric models.


\subsection{Markov Models}
\label{sec:ci_models}

Suppose $p(\mathbf{V})$ is a distribution over the set of all random variables of interest $\mathbf{V}$. If $\mathbf{X}, \mathbf{Y}, \mathbf{Z}$ are disjoint subsets of $\mathbf{V}$, a \textit{conditional independence} (CI) statement ``$\mathbf{X}$ and $\mathbf{Y}$ are independent given $\mathbf{Z}$ in $p$'', denoted $\mathbf{X} \Perp \mathbf{Y} \mid \mathbf{Z}$, means $p(\mathbf{x}, \mathbf{y} | \mathbf{z}) = p(\mathbf{x} | \mathbf{z}) p(\mathbf{y} | \mathbf{z})$ for all values of $\mathbf{X}, \mathbf{Y}, \mathbf{Z}$ such that $p(\mathbf{z}) > 0$.
Given a list $\mathscr{L}$ of CI statements, a statistical model associated with these constraints, written as $\{p(\mathbf{V}) : \text{ all CIs in }\mathscr{L} \text{ hold in }p(\bf{V})\}$ is called a \textit{Markov model}, or a \textit{conditional independence model}
\citep{Maathuis.Drton.ea.2018.HandbookGraphical}.
Important subclasses of Markov models are \textit{graphical models} associated with directed acyclic graphs (DAGs), undirected graphs (UGs), chain graph (CGs), and acyclic directed mixed graph (ADMGs) \citep{Maathuis.Drton.ea.2018.HandbookGraphical, Lauritzen.1996.GraphicalModels, Richardson.2003.MarkovProperties}.

Undirected graphs (UGs) are graphs with only undirected edges ($-$).
Directed acyclic graphs (DAGs) are graphs with only directed edges ($\rightarrow$) lacking directed cycles.  Bidirected graphs (BGs) are graphs with only bidirected edges ($\leftrightarrow$)
Chain graphs (CGs) are mixed graphs containing both directed and undirected edges with no partially directed cycles.  Acyclic directed mixed graphs (ADMGs) are mixed graphs with directed and bidirected edges lacking directed cycles (for the purposes of such cycles the presence of bidirected edges is ignored).  Simple examples of these graphs classes are shown in Fig.~\ref{fig:ci_models}.

\begin{figure}[H]
\centering
\scalebox{0.75}{
    \begin{tikzpicture}[>=stealth, node distance=1.3cm]
        \def\d{1.3cm}
        \begin{scope}
            \path[->, very thick]
            node[] (C) {$C$}
            node[right of=C] (D) {$D$}
            node[above of=C] (A) {$A$}
            node[above of=D] (B) {$B$}
            
            (A) edge[blue] (B)
            (A) edge[blue] (C)
            (B) edge[blue] (D)
            (C) edge[blue] (D)
            node[below of=C, yshift=\d/2, xshift=\d/2] () {
            (a) The two pathway DAG} 
            ;
        \end{scope}
        \begin{scope}[xshift=\d*3.5]
            \path[->, very thick]
            node[] (C) {$C$}
            node[right of=C] (D) {$D$}
            node[above of=C] (A) {$A$}
            node[above of=D] (B) {$B$}
            
            (A) edge[-, gray] (B)
            (A) edge[-, gray] (C)
            (B) edge[-, gray] (D)
            (C) edge[-, gray] (D)
            node[below of=C, yshift=\d/2, xshift=\d/2] () {
            (b) The undirected square} 
            ;
        \end{scope}
        \begin{scope}[yshift=-\d*2.5]
            \path[->, very thick]
            node[] (C) {$C$}
            node[right of=C] (D) {$D$}
            node[above of=C] (A) {$A$}
            node[above of=D] (B) {$B$}
            
            (A) edge[blue] (C)
            (B) edge[blue] (D)
            (C) edge[-, gray] (D)
            node[below of=C, yshift=\d/2.5, xshift=\d/2, text width=\d*2.2] () {(c) The dyadic minimal complex CG} 
            ;
        \end{scope}
        \begin{scope}[yshift=-\d*2.5, xshift=\d*3.5]
            \path[->, very thick]
            node[] (C) {$C$}
            node[right of=C] (D) {$D$}
            node[above of=C] (A) {$A$}
            node[above of=D] (B) {$B$}
            
            (A) edge[blue] (C)
            (B) edge[blue] (D)
            (C) edge[<->, red] (D)
            node[below of=C, yshift=\d/2, xshift=\d/2] () {(d) The Bell scenario} 
            ;
        \end{scope}
        \begin{scope}[yshift=-\d*5]
            \path[->, very thick]
            node[] (C) {$C$}
            node[right of=C] (D) {$D$}
            node[above of=C] (A) {$A$}
            node[above of=D] (B) {$B$}
            
            (A) edge[<->, red] (C)
            (A) edge[<->, red] (B)
            (B) edge[<->, red] (D)
            (C) edge[<->, red] (D)
            node[below of=C, yshift=\d/2, xshift=\d/2] () {(e) The bidirected square} 
            ;
        \end{scope}
        \begin{scope}[yshift=-\d*5, xshift=\d*3.5]
            \path[->, very thick]
            node[] (C) {$C$}
            node[right of=C] (D) {$D$}
            node[above of=C] (A) {$A$}
            node[above of=D] (B) {$B$}
            
            (A) edge[blue] (B)
            (A) edge[blue] (C)
            (A) edge[blue] (D)
            (B) edge[blue] (D)
            (B) edge[blue] (C)
            (C) edge[blue] (D)
            node[below of=C, yshift=\d/2, xshift=\d/2] () {(f) A complete $4$ vertex DAG}
            ;
        \end{scope}
    \end{tikzpicture}
}
\caption{Six different graphs used to define six different graphical Markov models. Models 
associated with graphs in (b), (c), (d), and (e) are not equivalent to any DAG model, so their tangent spaces cannot be characterized using the technique in Section~\ref{sec:dag_tangent}.}
\label{fig:ci_models}
\end{figure}
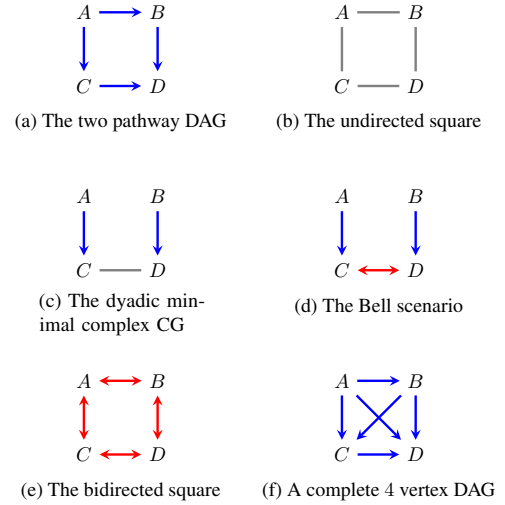

Graphical models use graphical separation criteria to encode CI constraints. The (finite) set of all CIs that hold in a graphical model is given by a global Markov property, while local Markov properties provide a small list of CIs that logically imply all others.
Below, we define local Markov properties for models associated with UGs, DAGs, CGs, bidirected graphs (BGs), and ADMGs, and refer reader to \citet{Maathuis.Drton.ea.2018.HandbookGraphical, Richardson.2003.MarkovProperties} for additional details.

To describe the local Markov properties, we need the following basic graphical concepts. If $\mathcal{G}$ is an UG, denote $\nei_\mathcal{G}(V) \equiv \{Z \in \mathbf{V}: Z - V \text{ in }\mathcal{G}\}$ as the set of neighbors of $V$ in $\mathcal{G}$. If $\mathcal{G}$ is a DAG, denote $\pa_\mathcal{G}(V) \equiv \{Z \in \mathbf{V}: Z \rightarrow V \text{ in }\mathcal{G}\}$ as the set of parents of $V$ in $\mathcal{G}$, $\ch_\mathcal{G}(V) \equiv \{Z \in \mathbf{V}: V \rightarrow Z \text{ in }\mathcal{G}\}$ as the set of children of $V$ in $\mathcal{G}$, $\de_\mathcal{G}(V) \equiv \{Z \in \mathbf{V}: V \rightarrow \cdots \rightarrow Z \text{ in }\mathcal{G}\}$ as the set of descendants of $V$ in $\mathcal{G}$,
$\an_\mathcal{G}(V) \equiv \{Z \in \mathbf{V}: Z \rightarrow \cdots \rightarrow V \text{ in }\mathcal{G}\}$ as the set of ancestors of $V$ in $\mathcal{G}$,
and $\nd_\mathcal{G}(V) \equiv \mathbf{V} \setminus \de_\mathcal{G}(V)$ as the set of non-descendants of $V$ in $\mathcal{G}$. Note that by convention, $V \in \de_{\cal G}(V) \cap \an_{\cal G}(V)$.  For sets of variables, the sets of parents, ancestors and descendants are defined disjunctively.  That is $\pa_{\cal G}({\bf A}) \equiv \bigcup_{A \in {\bf A}} \pa_{\cal G}(A)$, $\an_{\cal G}({\bf A}) \equiv \bigcup_{A \in {\bf A}} \an_{\cal G}(A)$ and $\de_{\cal G}({\bf A}) \equiv \bigcup_{A \in {\bf A}} \de_{\cal G}(A)$.

For CGs, define $\pa_\mathcal{G}(V)$ as in DAGs and $\nei_{\cal G}(V)$ as in UGs. The set of descendants of $V$, $\de_{\cal G}(V)$ is defined as the set of all variables $Z$ with a \emph{partially directed path} from $V$ to $Z${ -- a sequence of distinct vertices such that there is an edge ($-$ or $\rightarrow$) between any pair of consecutive vertices, with all directed edges pointing in the same direction.}  As before, $V \in \de_{\cal G}(V)$ and $\nd_{\cal G}(V) = {\bf V} \setminus \de_{\cal G}(V)$.
We further define the \emph{boundary of $V$}, $\bd_{\cal G}(V)$ as $\nei_{\cal G}(V) \cup \pa_{\cal G}(V)$.
Finally, we define a chain component in a CG $\mathcal{G}$ to be a connected component in the edge subgraph of $\mathcal{G}$ retaining only undirected edges.

If $\mathcal{G}$ contains bidirected edges, denote $\sib_\mathcal{G}(V) \equiv \{Z \in \mathbf{V}: Z \leftrightarrow V \text{ in }\mathcal{G}\}$ as the set of 
siblings of $V$ in $\mathcal{G}$. Furthermore, a subset $\mathbf{C} \subseteq \mathbf{V}$ is called a bidirected connected subset if it is a connected component in the edge subgraph of $\mathcal{G}$ retaining only bidirected edges. A district is a maximal bidirected connected subset, and $\dis_\mathcal{G}(V)$ denotes the district in $\mathcal{G}$ containing $V$.

For ADMGs, all mentioned graphical concepts for DAGs and BGs apply, and define $\mb_{\mathcal{G}}(V) = \pa_{\mathcal{G}}(\dis_{\mathcal{G}}(V)) \cup \big( \dis_{\mathcal{G}}(V) \setminus \{ V \} \big)$ as the Markov blanket of $V$ in $\mathcal{G}$. Furthermore, a subset $\mathbf{A}$ is called ancestral if
whenever $V \in \an_{\cal G}({\bf A})$ then $V \in {\bf A}$.

Given a graph ${\cal G}$ with a vertex set ${\bf V}$, and ${\bf A} \subseteq {\bf V}$, define ${\cal G}_{\bf A}$, the \emph{subgraph of ${\cal G}$ induced by ${\bf A}$}, to be the subgraph of ${\cal G}$ containing only the vertex set ${\bf A}$ and edges among $\mathbf{A}$.

The list of CIs corresponding to the local Markov property for UGs, DAGs, CGs, BGs and ADMGs with the vertex set ${\bf V}$ are as follows
\begin{itemize}[leftmargin=*]
    \item UGs: $\big( V \Perp \mathbf{V} \setminus \big( \{V\} \cup \nei_\mathcal{G}(V) \big) \big| \nei_\mathcal{G}(V) \big)$, for all $V \in \mathbf{V}$.
    \item DAGs: $\big( V \Perp \nd_\mathcal{G}(V) \setminus \pa_\mathcal{G}(V) \big| \pa_\mathcal{G}(V) \big)$, for all $V \in \mathbf{V}$.
    \item CGs: $\big( V \Perp \nd_\mathcal{G}(V) \setminus \bd_\mathcal{G}(V) \big| \bd_\mathcal{G}(V) \big)$,
    for all $V \in \mathbf{V}$.
    \item BDs:
    $\big( V \Perp {\bf A} \setminus \dis_{{\cal G}_{\bf A}}(V) \big| \dis_{{\cal G}_{\bf A}}(V)\big)$, for all $V$ and all subsets ${\bf A}$ of ${\bf V}$ such that $V \in {\bf A}$.
    \item ADMGs: $\big( V \Perp \mathbf{A} \setminus \big( \mb_{\mathcal{G}_{\bf A}}(V) \cup \{V\} \big) \big| \mb_{\mathcal{G}_{\bf A}}(V) \big)$, for all $V$ and ancestral set $\mathbf{A}$ such that $V \in \mathbf{A} \subseteq \nd_{\mathcal{G}}(V)$.
\end{itemize}

Note that, as expected, the local properties for the UG and DAG model are special cases of the local property for the CG model, and the local properties for the BG and DAG model are special cases of the local property for the ADMG model.


For example, the following are the lists of CIs for the models associated with the graphs in Figure~\ref{fig:ci_models}.
\begin{enumerate}[nosep, label=(\alph*)]
    \item The two pathways DAG model:\\
    $C \Perp B \mid A$ and $D \Perp A \mid B,C$.
    \item The undirected square UG model:\\
    $C \Perp B \mid A,D$ and $D \Perp A \mid B,C$.
    \item The dyadic minimal complex CG model:\\
    $A\Perp B$, $C \Perp B \mid A,D$ and $D \Perp A \mid B,C$.
    \item The Bell scenario ADMG model:\\
    $A \Perp B, D$ and $B \Perp A, C$.
    \item The bidirected square scenario ADMG model:\\
    $A \Perp D$ and $B \Perp C$.
    \item The complete DAG: the empty list, as the model is saturated.
\end{enumerate}

{
Note that there are graphical models associated with ADMGs defined not only in terms of CIs but also generalized CIs, sometimes called ``Verma constraints'' \citep{Richardson.Evans.ea.2023.NestedMarkov}.  As an example, the model associated with Fig.~\ref{fig:verma_graph} is defined by a CI: $C \Perp A \mid B$ and a generalized independence constraint stating that $\sum_B p(D \mid C, B, A) p(B \mid A)$ is a function of only $C$ and $D$.
In general, both CIs and Verma constraints are needed to describe marginal models for hidden variable DAGs.}
For the purposes of this work, we only consider ADMG models defined exclusively via ordinary CI constraints, called \emph{ordinary Markov models} \citep{Richardson.2003.MarkovProperties,evans18margins}.

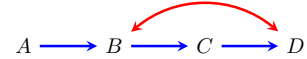
\begin{figure}[H]
\centering
\scalebox{0.8}{
    \begin{tikzpicture}[>=stealth, node distance=1.5cm]
        \def\d{1.5cm}
        \begin{scope}
            \path[->, very thick]
            node[xshift=-\d] (A) {$A$}
            node[right of=A] (B) {$B$}
            node[right of=B] (C) {$C$}
            node[right of=C] (D) {$D$}
            
            (A) edge[blue] (B)
            (B) edge[blue] (C)
            (C) edge[blue] (D)
            (B) edge[<->, red, bend left=40] (D)
            ;
        \end{scope}
    \end{tikzpicture}
}
\caption{A graph representing a model with a Verma constraint.
}
\label{fig:verma_graph}
\end{figure}

\subsection{Semi-parametric Inference}
\label{sec:semiparam}

As we further describe in Sections \ref{sec:dag_tangent} and \ref{sec:ci_tangent}, there is a close relationship between Markov model restrictions and the tangent space of the model.  This relationship is used to construct estimators via theory of semi-parametric models which we now briefly outline. For further details, we refer the reader to excellent tutorials in \citet{Kennedy.2024.SemiparametricDoubly} and \citet{Tsiatis.2006.SemiparametricTheory}.

We refer to a semi-parametric model as a set of probability densities $\mathscr{P} = \{p(\mathbf{v}; \theta) : \theta \in \Theta\}$ relative to a measure $\mu$, parameterized by an infinite dimensional nuisance parameters $\theta$, with the true data generating distribution $p_0 \in \mathscr{P}$.
As an example, the semi-parametric model associated with the DAG model in Fig.~\ref{fig:ci_models}a is $\mathscr{P}^{(a)} = \{p(a,b,c,d) : p(a,b,c,d) = p(a) p(b \mid a) p(c \mid a) p(d \mid b,c) \}$.

The \emph{contaminated}
parametric submodel is a smooth parametric model $\mathscr{P}_{\varepsilon} = \{p_{\varepsilon}(\mathbf{v}) : \varepsilon \in \mathcal{E} \subseteq \mathbb{R}\}$ defined around the true distribution $p_0$ such that $p_{\varepsilon=0} = p_0$, and is contained in the semi-parametric model $\mathscr{P}_{\varepsilon} \subseteq \mathscr{P}$. In the subsequent discussion, we will refer to $\mathscr{P}_{\varepsilon}$ as the parametric submodel, following \citet{Tsiatis.2006.SemiparametricTheory}.
We also denote the saturated (unrestricted) model as $\mathscr{P}^{\text{all}}$. As an example, the model associated with the complete DAG in Fig.~\ref{fig:ci_models}f is saturated.
By definition, all semiparametric models are contained in $\mathscr{P}^{\text{all}}$.

For a parametric submodel $\mathscr{P}_{\varepsilon}$, the score function is the derivative of the log likelihood at the truth $s(\mathbf{v}) = \frac{\partial}{\partial \varepsilon} \log p_{\varepsilon}(\mathbf{v})|_{\varepsilon = 0}$. This function is an element of the Hilbert space $\mathcal{H}$ of all mean-zero $L^2(p_0)$ functions $f(\mathbf{v})$ in which the inner product is relative to the truth distribution $\langle f, g \rangle = \int f(\mathbf{v}) g(\mathbf{v}) p_0(\mathbf{v}) d \mathbf{v}, \forall f, g \in \mathcal{H}$. The tangent space of a semi-parametric model, denoted $\mathcal{T}$, is the closure of all parametric submodels' score functions, thus ${\cal T}$ is a subspace of $\mathcal{H}$. For the saturated model $\mathscr{P}^{\text{all}}$, $\mathcal{T} = \mathcal{H}$. The orthogonal complement of the tangent space, denoted $\mathcal{T}^\perp$, is the subspace of all functions in $\mathcal{H}$ orthogonal to $\mathcal{T}$, so $\langle f, g \rangle = 0$ for all $f \in \mathcal{T}^\perp$ and $g \in \mathcal{T}$.

The target of inference in a semi-parametric model is a smooth function $\psi: \mathscr{P} \mapsto \mathbb{R}^d$ where $d$ is a fixed integer. For simplicity, we will only consider scalar target, so $d=1$.  For example, an adjustment formula functional arising in causal inference applications for the examples in Figure~\ref{fig:ci_models} yields a scalar parameter $\psi(p) = \mathbb{E}[\mathbb{E}[D \mid A = a_0, B, C]]$.  The main result can be extended easily to $d > 1$ but the notation is more involved.

In semi-parametric theory we are interested in the class of \emph{regular and asymptotically linear (RAL)} estimators for $\psi$. Given data $\mathbf{V}_1, \ldots, \mathbf{V}_n$ drawn from a distribution $p \in \mathscr{P}$, an estimator $\hat{\psi}_n$ for $\psi$ is called asymptotically linear if there is a function $\phi \in \mathcal{H}$ such that
\begin{equation}
\label{eq:asymp_linear}
    \sqrt{n} (\hat{\psi}_n - \psi) = \frac{1}{\sqrt{n}} \sum_{i=1}^n \phi(\mathbf{V}_i) + o_p(1).
\end{equation}
By the central limit theorem, the limiting distribution of such an estimator $\hat{\psi}_n$ is normal with root-$n$ rate of convergence,
\begin{equation}
\label{eq:asymp_linear_clt}
    \sqrt{n} (\hat{\psi}_n - \psi) \rightsquigarrow \mathcal{N}(0, \mathbb{E}[\phi^2]),
\end{equation}
where $\rightsquigarrow$ denotes convergence in distribution. Estimators exhibit asymptotically linear property locally uniformly around the truth are called regular and asymptotically linear (RAL) estimators, which excludes super efficient estimators like the Hodges estimator. The function $\phi$ is called an \textit{influence function for the estimator} $\hat{\psi}_n$. RAL estimators are attractive because they are root-$n$ consistent and asymptotically normal, and the variance equal the variance of its influence function.

A closely related notion is the \textit{influence functions for the target parameter}.
Specifically when the distribution changes from $p$ to $\bar{p}$, one can approximate the change in the target functional using the following \emph{von Mises expansion}
\begin{equation}
    \psi(\bar{p}) - \psi(p) = \int \varphi(\mathbf{v}) (\bar{p} - p)(\mathbf{v}) d \mathbf{v} + R[(\bar{p} - p)^2].
\end{equation}
The integral is the first-order change of the target parameter with respect to the change in the distribution,
while $R$ is the second-order term of the expansion. The function $\varphi \in \mathcal{H}$ characterizes the first-order derivative of the target functional, and is referred to as the influence function of the target parameter $\psi$.
The von Mises expansion implies the following for all parametric submodels $p_{\varepsilon}$ with score function $s \in \mathcal{T}$
\begin{equation}
\label{eq:path_dev}
    \frac{\partial}{\partial \varepsilon} \psi(p_{\varepsilon}) \bigg|_{\varepsilon=0} = \int \varphi(\mathbf{v}) s(\mathbf{v}) p_0(\mathbf{v}) d \mathbf{v},
\end{equation}
where, under regularity conditions, it suffices to consider only parametric submodels of the form $p_{\varepsilon}(\mathbf{v}) = p_{0}(\mathbf{v})(1 + \varepsilon s(\mathbf{v}))$.
Equation~\ref{eq:path_dev} is called the \textit{pathwise differentiability} condition, which is implied by the Neyman orthogonality condition \citep{Chernozhukov.Chetverikov.ea.2018.DoubleDebiased}, and is central to the derivation of target's influence function. In practice, one first obtains a target parameter's influence function $\varphi$ from Equation~\ref{eq:path_dev}, then from it constructs a RAL estimator, whose influence function coincides with the target's influence function, $\phi = \varphi$. We focus on the target parameter's influence functions and will refer to them as just influence functions (IFs).

From Equation~\ref{eq:path_dev}, it is evident that if $\varphi$ is an influence function and $f \in \mathcal{T}^\perp$, then $\varphi + f$ is also an influence function, since $\langle f, s \rangle = \int f(\mathbf{v}) s(\mathbf{v}) p_0(\mathbf{v}) d \mathbf{v} = 0$. Therefore, the class of all influence functions is $\varphi \oplus \mathcal{T}^\perp$, where $\varphi$ is any solution of Equation~\ref{eq:path_dev} and $\oplus$ denotes the direct sum (\citet{Tsiatis.2006.SemiparametricTheory}, Theorem 4.3). Thus, knowing the tangent space's orthogonal complement $\mathcal{T}^\perp$ allows one to derive the class of influence functions for any target of inference.


\section{The orthocomplement of the Tangent Space of a DAG Model}
\label{sec:dag_tangent}

The tangent space of a DAG model and its orthogonal complement have been derived in multiple contexts, see Theorem 4.5 in \citet{Tsiatis.2006.SemiparametricTheory}, Lemma 24 in \citet{Rotnitzky.Smucler.2020.EfficientAdjustment}, and Lemma 3 in \citet{Bhattacharya.Nabi.ea.2022.SemiparametricInference}. In the interests of being self-contained, we repeat the derivation using examples in Figure~\ref{fig:ci_models}a in order to highlight the main techniques. Despite using specific examples, the derivation is without loss of generality. Subsequently, we explain the reason why these techniques fail to apply to Markov models beyond DAG models.

\subsection{
A Worked Example
}
\label{sec:dag_tangent:derivation}

Consider the DAG in Figure~\ref{fig:ci_models}a, whose semi-parametric model is $\mathscr{P}^{(a)}$. To find the 
tangent space $\mathcal{T}$, we first derive the score function for an arbitrary parametric submodel $p_{\varepsilon}$, then take the closure. It turns out that $\mathcal{T}$ is the range of a projection operator $\Pi$, hence the orthogonal complement of $\mathcal{T}$ is simply $h - \Pi(h)$ for all $h \in \mathcal{H}$. These steps are demonstrated in detail below.

\paragraph{The score function of a parametric submodel.}
Consider a parametric submodel $\mathscr{P}_{\varepsilon}$. By the DAG factorization property, elements of $\mathscr{P}_{\varepsilon}$ satisfy
\begin{equation}
\label{eq:dag_llh}
    p_{\varepsilon}(a,b,c,d) = p_{\varepsilon}(a) p_{\varepsilon}(b \mid a) p_{\varepsilon}(c \mid a) p_{\varepsilon}(d \mid c,b).
\end{equation}
This implies that the total score function is the sum
\begin{equation}
\label{eq:dag_score_1}
    s(a,b,c,d) = s(a) + s(b \mid a) + s(c \mid a) + s(d \mid c,b),
\end{equation}
where, for any pair of disjoint subsets $\mathbf{X}, \mathbf{Y} \subseteq \mathbf{V}$, a component score function $s(\mathbf{x} \mid \mathbf{y})$ is defined as $\frac{\partial  \log p_\varepsilon(\mathbf{x} \mid \mathbf{y})}{\partial \varepsilon} |_{\varepsilon=0}$.

Since a component score function $s(\mathbf{x} \mid \mathbf{y})$ has the property that its conditional mean is zero $\int s(\mathbf{x} \mid \mathbf{y}) p_0(\mathbf{x} \mid \mathbf{y}) d \mathbf{x} = 0$, the component scores are in the following subspaces of $\mathcal{H}$
\begin{align}
\label{eq:dag_subtangent}
    s(a) \in \mathcal{T}_A & \equiv \{\mathbb{E}[h | a] : \forall h \in \mathcal{H}\},
    \\
    \notag
    s(b \mid a) \in \mathcal{T}_{B \mid A} & \equiv \{\mathbb{E}[h | b, a] - \mathbb{E}[h | a] : \forall h \in \mathcal{H}\},
    \\
    \notag
    s(c \mid a) \in \mathcal{T}_{C \mid A} & \equiv \{\mathbb{E}[h | c, a] - \mathbb{E}[h | a] : \forall h \in \mathcal{H}\},
    \\
    \notag
    s(d \mid c,b) \in \mathcal{T}_{D \mid CB} & \equiv \{\mathbb{E}[h | d, c, b] - \mathbb{E}[h | c, b] : \forall h \in \mathcal{H}\}.
\end{align}
The expectations are relative to $p_0$. In the Appendix, we show that these subspaces are pairwise orthogonal by checking that the inner products of their elements vanish.

\paragraph{The model's tangent space.} The semi-parametric tangent space $\mathcal{T}$ is the closure of all parametric submodels's score functions, hence it must be a closed subspace of the direct sum $\mathcal{T}_A \oplus \mathcal{T}_{B|A} \oplus \mathcal{T}_{C|A} \oplus \mathcal{T}_{D|BC}$. In fact, it is actually this direct sum
\begin{equation}
\label{eq:dag_tangent}
\begin{aligned}
    \mathcal{T}
    = & \mathcal{T}_A \oplus \mathcal{T}_{B|A} \oplus \mathcal{T}_{C|A} \oplus \mathcal{T}_{D|BC}.
\end{aligned}
\end{equation}
We refer to the subspaces on the r.h.s as tangent subspaces.
Any element of $\mathcal{T}$ can be explicitly constructed as the sum of four elements of the four tangent subspaces, whose explicit forms are given by Equation~\ref{eq:dag_subtangent}.
A detailed proof of Equation~\ref{eq:dag_tangent} is provided in the Appendix.
In essence, the equality is established once it is shown that each tangent subspace is contained in $\mathcal{T}$.
To illustrate, consider showing that $\mathcal{T}_{B|A} \subseteq \mathcal{T}$.
Pick any $f \in \mathcal{T}_{B|A}$ and define the parametric model $p_{\varepsilon} = p_{0}(1 + \varepsilon f)$, where $\varepsilon$ is small enough so that $p_{\varepsilon} \geq 0$. The score function at $p_{0}$ is $\frac{\partial}{\partial \varepsilon} \log p_{\varepsilon}|_{\varepsilon=0} = f$.
This parametric model contains the true distribution $p_0$ at $\varepsilon = 0$, and it satisfies all the DAG constraints: $p_{\varepsilon}(c \mid b, a) = p_{\varepsilon}(c \mid a)$ and $p_{\varepsilon}(d \mid c, b, a) = p_{\varepsilon}(d \mid c, b)$.
Therefore, $p_{\varepsilon}$ defines a valid parametric submodel of $\mathscr{P}$, confirming $\mathcal{T}_{B|A} \subseteq \mathcal{T}$.

\paragraph{The orthogonal complement of the model tangent space.}
If $\mathcal{S}$ is a closed subspace of $\mathcal{H}$, the \textit{orthogonal complement} of $\mathcal{S}$ in $\mathcal{H}$, denoted $\mathcal{S}^\perp$, is defined as the set of all residues $h - \Pi(h \mid \mathcal{S})$ for all $h \in \mathcal{H}$, where $\Pi(\cdot \mid \mathcal{S})$ denotes the projection operator onto the subspace $\mathcal{S}$. Therefore, we first derive $\Pi(\cdot \mid \mathcal{T})$ before constructing $\mathcal{T}^\perp$.

The expressions of $\mathcal{T}_A, \mathcal{T}_{B|A}, \mathcal{T}_{C|A}, \mathcal{T}_{D|CB}$ in Equation~\ref{eq:dag_subtangent}
show that these subspaces are ranges of corresponding linear operators $\Lambda_A, \Lambda_{B|A}, \Lambda_{C|A}, \Lambda_{D|CB}$
appeared in these expressions, respectively. For example, the linear operator $\Lambda_{C|A}$ is defined by $\Lambda_{C|A}(h)(c,a) \equiv \mathbb{E}[h \mid c, a] - \mathbb{E}[h \mid a]$, and it maps the Hilbert space $\mathcal{H}$ into $\mathcal{T}_{C|A}$.
In fact, these operators are actually projection operators onto these subspaces,
which can be verified directly by checking the inner product between the residue and the image. For example, in the case of $\mathcal{T}_{C|A}$, we have $\langle h - \Lambda_{C|A}(h), \Lambda_{C|A}(h) \rangle = \mathbb{E} \big[ (h - \mathbb{E}[h \mid c, a] + \mathbb{E}[h \mid a]) (\mathbb{E}[h \mid c, a] - \mathbb{E}[h \mid a]) \big] = 0$ for all $h \in \mathcal{H}$.
Since the subspaces $\mathcal{T}_A, \mathcal{T}_{B|A}, \mathcal{T}_{C|A}, \mathcal{T}_{D|CB}$ are pairwise orthogonal, the projection of a function $h \in \mathcal{H}$ onto the tangent space $\mathcal{T}$ is given by the sum of the individual projections
\begin{equation}
\label{eq:dag_tangent_proj}
\begin{aligned}
    \Pi(h \mid \mathcal{T})
    = & \Lambda_{A}(h) + \Lambda_{B|A}(h) + \Lambda_{C|A}(h) + \Lambda_{D|CB}(h)
    \\
    = & \mathbb{E}[h \mid d, c, b] - \mathbb{E}[h \mid c, b] + \mathbb{E}[h \mid c, a] \\
    & + \mathbb{E}[h \mid b, a] - \mathbb{E}[h \mid a].
\end{aligned}
\end{equation}
Indeed, one can directly verify that this is the expression of the projection operator onto $\mathcal{T}$ by checking the inner product $\langle h - \Pi(h \mid \mathcal{T}), \Pi(h \mid \mathcal{T}) \rangle = 0$ for all $h \in \mathcal{H}$. Hence, the orthogonal complement of the tangent space is
\begin{equation}
\label{eq:dag_orthocomp}
\begin{aligned}
    \mathcal{T}^{\perp}
    = \{ & h - \mathbb{E}[h \mid d, c, b] + \mathbb{E}[h \mid c, b] - \mathbb{E}[h \mid c, a] \\
    & - \mathbb{E}[h \mid b, a] + \mathbb{E}[h \mid a] : \forall h \in \mathcal{H} \}.
\end{aligned}
\end{equation}

Given any DAG model, these described steps can be used to derive the tangent space and its orthogonal complement, as done in \citet{Tsiatis.2006.SemiparametricTheory, Rotnitzky.Smucler.2020.EfficientAdjustment, Bhattacharya.Nabi.ea.2022.SemiparametricInference}. In the special case when the DAG model has only one constraint, we 
reproduce the following well-known lemma.
\begin{lemma}
    \label{lm:single_dag_model_orthocomp}
    Suppose $\mathscr{P}$ is a single independence DAG model over variables $\mathbf{V}$, defined by a single constraint $\mathbf{X} \Perp \mathbf{Y} \mid \mathbf{Z}$. The model's tangent space $\mathcal{T}$ and its orthogonal complement $\mathcal{T}^\perp$ are, respectively
    \begin{equation}
    \begin{aligned}
        \label{eq:single_dag_model_orthocomp}
        \mathcal{T} = & \big\{ h - \mathbb{E}[h \mid \mathbf{x}, \mathbf{y}, \mathbf{z}]
        + \mathbb{E}[h \mid \mathbf{x}, \mathbf{z}] - \mathbb{E}[h \mid \mathbf{z}]
        \\
        & + \mathbb{E}[h \mid \mathbf{y}, \mathbf{z}]: \forall h \in \mathcal{H} \big\},
        \\
        \mathcal{T}^\perp = & \big\{ \mathbb{E}[h \mid \mathbf{x}, \mathbf{y}, \mathbf{z}] - \mathbb{E}[h \mid \mathbf{y}, \mathbf{z}]
        \\
        & - \mathbb{E}[h \mid \mathbf{x}, \mathbf{z}] + \mathbb{E}[h \mid \mathbf{z}]
        : \forall h \in \mathcal{H} \big\}.
    \end{aligned}
    \end{equation}
\end{lemma}
The proof of this lemma, which is in the Appendix, is similar to the above worked example, except that the model has only one constraint instead of two constraints, and disjoint subsets of variables replace the univariate variables.

\subsection{Methodological obstacles for obtaining general tangent spaces}
\label{ssec:implicit_form}

In the worked example in the previous section, and indeed in all DAG models, both the tangent space $\mathcal{T}$ and its orthogonal complement $\mathcal{T}^\perp$ are expressible in closed form.  That is, one can construct arbitrary elements of $\mathcal{T}$ and $\mathcal{T}^\perp$ using Equation~\ref{eq:dag_tangent} and \ref{eq:dag_orthocomp}, respectively. This is particularly useful for inference, since if $\varphi$ is a known influence function for some target parameter $\psi(p)$, then
\begin{equation}
\begin{aligned}
    \varphi + h - \mathbb{E}[h \mid d, c, b] + \mathbb{E}[h \mid c, b] - \mathbb{E}[h \mid c, a] & \\
     - \mathbb{E}[h \mid b, a] + \mathbb{E}[h \mid a]
\end{aligned}
\end{equation}
is also an influence function of that target, where $h$ is an arbitrary function in $\mathcal{H}$. The key reason for the existence of the closed form expressions for $\mathcal{T}$ and $\mathcal{T}^\perp$ is the fact that the DAG model likelihood in Equation~\ref{eq:dag_llh} is also in closed form, as a single factorization representing all constraints in the local Markov property of the model.

On the other hand, in Markov models associated with ADMGs, no single factorization representing all model constraints is known for general state spaces.  For example, in the Bell scenario in Figure~\ref{fig:ci_models}d, a closed form factorization of $p_{\varepsilon}(a,b,c,d)$ representing both constraints $A \Perp B, D$ and $B \Perp A, C$ is unknown in general settings.
Instead, distributions satisfying both constraints may be obtained iteratively by solving the following system of
two equations
\begin{align}
    p_{\varepsilon}(a,b,c,d) &= p_{\varepsilon}(a) p_{\varepsilon}(b, d) p_{\varepsilon}(c | d,b,a), \label{eq:bell_llh_p1}
    \\
    p_{\varepsilon}(a,b,c,d) &= p_{\varepsilon}(b) p_{\varepsilon}(c, a) p_{\varepsilon}(d | c,b,a). \label{eq:bell_llh_p2}
\end{align}
The first equation forces the distribution to satisfy $A \Perp B, D$, while the second equation forces the likelihood to satisfy $B \Perp A, C$.
We will refer to them as the likelihood constraint equations. 
Mathematical objects defined as solutions of some system of equations, such as $p_{\varepsilon}$ above, is said to have \textit{implicit expressions}.
Subsequently, the score function of this parametric submodel is also defined implicitly as solutions to the following two equations
\begin{align}
    s(a,b,c,d) &= s(a) + s(b, d) + s(c | a,b,d), \label{eq:bell_score_p1}
    \\
    s(a,b,c,d) &= s(b) + s(a, c) + s(d | b,a,c), \label{eq:bell_score_p2}
\end{align}
which are referred to as the score constraint equations.
To see that these two equation solves for the same object, we can equivalently reexpress these component score functions in terms of the total score function $s(a,b,c,d)$ as
\begin{align}
    \mathbb{E}[s \mid a,b,d] &= \mathbb{E}[s \mid a] + \mathbb{E}[s \mid b,d], \label{eq:bell_total_score_p1}
    \\
    \mathbb{E}[s \mid b,a,c] &= \mathbb{E}[s \mid b] + \mathbb{E}[s \mid a, c]. \label{eq:bell_total_score_p2}
\end{align}

Similarly to Equations~\ref{eq:bell_llh_p1} and \ref{eq:bell_llh_p2}, solving this system of equations is an open problem in general settings. These problems are closely related to implicitization problems in algebraic geometry \citep{Cox.Little.ea.2015.IdealsVarieties} and parameterization problems of graphical model likelihoods \citep{Evans.Richardson.2019.SmoothIdentifiable, Shpitser.2023.LauritzenChenLikelihood},
which may be viewed as the problem of transforming an implicit definition of a likelihood surface via constraints into an explicit definition of said likelihood indexed by variationally independent parameters.

Absent such a likelihood definition, tangent spaces for graphical models such as the Bell scenario model will have an implicit expression of the form:
$\mathcal{T} = \{\text{all common solutions of Equations~\ref{eq:bell_total_score_p1} and \ref{eq:bell_total_score_p2}}\}$. It is unknown how to construct the projection operator onto this subspace. Consequently, the orthogonal complement for the tangent space also has an implicit form: $\mathcal{T}^\perp = \{\text{all } f \text{ orthogonal to all common solutions of \ref{eq:bell_total_score_p1} and \ref{eq:bell_total_score_p2}}\}$. In other words, as opposed to the DAG case, elements of $\mathcal{T}^\perp$ cannot be written down in closed form.


In the next section, we illustrate an alternative simple approach for obtaining the characterization of the orthogonal complement of the tangent space that bypasses the need for an explicit representation of the likelihood or the tangent space.


\section{The orthocomplement of the Tangent Space of a CI Model}
\label{sec:ci_tangent}

A general Markov model may be viewed as an intersection of single independence models, each with one CI constraint. Thus, the tangent space of the general Markov model is the intersection of all the single independence models's tangent spaces. Consequently, the orthogonal complement of this intersection is the span of all the orthogonal complements of the single independence models's tangent spaces. This geometric characterization of the tangent space's orthogonal complement for intersection models first appears in Lemma 1.7 of \citet{VanDerLaan.Robins.2003.UnifiedMethods}; our argument is an 
application of this observation in the context of Markov models.

\subsection{A worked example: the Bell scenario model}
\label{sec:ci_tangent:bell}

We define the following two models with only one constraint
\begin{align}
    \mathscr{P}_1 & = \{\text{all }p(a,b,c,d) \text{ satisfying } A \Perp B, D \}, \label{eq:bell_p1}
    \\
    \mathscr{P}_2 & = \{\text{all }p(a,b,c,d) \text{ satisfying } B \Perp A, C \}. \label{eq:bell_p2}
\end{align}
The Bell scenario model satisfying both constraints is then the intersection of these two models,
\begin{equation}
\label{eq:bell_model}
    \mathscr{P}^{(d)} = \mathscr{P}_1 \cap \mathscr{P}_2.
\end{equation}
Graphical models are commonly defined in this way, that is as intersections of simpler models 
\citep{Lauritzen.1996.GraphicalModels, Richardson.2003.MarkovProperties,Richardson.Evans.ea.2023.NestedMarkov}.

Both models $\mathscr{P}_1$ and $\mathscr{P}_2$ are \textit{single independence DAG models} -- DAG models defined by a single constraint.
Consequently, the technique described in Section~\ref{sec:dag_tangent} applies to these models. The parametric submodel likelihoods for $\mathscr{P}_1$ and $\mathscr{P}_2$ are Equations~\ref{eq:bell_llh_p1} and \ref{eq:bell_llh_p2}, respectively. The score functions for these parametric submodels are given by Equations~\ref{eq:bell_score_p1} and \ref{eq:bell_score_p2} (equivalently, Equations~\ref{eq:bell_total_score_p1} and \ref{eq:bell_total_score_p2}), respectively. By Lemma~\ref{lm:single_dag_model_orthocomp}, the tangent spaces for these models are
\begin{equation}
\begin{aligned}
    \mathcal{T}_1 & = \{ h - \mathbb{E}[h | a, b, d] + \mathbb{E}[h | a] + \mathbb{E}[h | b, d] : \forall h \in \mathcal{H} \},
    \\
    \mathcal{T}_2 & = \{ h - \mathbb{E}[h | b, a, c] + \mathbb{E}[h | b] + \mathbb{E}[h | a, c] : \forall h \in \mathcal{H} \},
\end{aligned}
\end{equation}
and their orthogonal complements are
\begin{equation}
\begin{aligned}
    \mathcal{T}_1^{\perp} & = \{ \mathbb{E}[h | a, b, d] - \mathbb{E}[h | a] - \mathbb{E}[h | b, d] : \forall h \in \mathcal{H} \}, \\
    \mathcal{T}_2^{\perp} & = \{ \mathbb{E}[h | b, a, c] - \mathbb{E}[h | b] - \mathbb{E}[h | a, c] : \forall h \in \mathcal{H} \}.
\end{aligned}
\end{equation}

The tangent space $\mathcal{T}$ of the Bell scenario model is the closure of all functions satisfying both Equations~\ref{eq:bell_total_score_p1} and \ref{eq:bell_total_score_p2}, $\mathcal{T} = \overline{\mathcal{T}_1 \cap \mathcal{T}_2}$. Since both $\mathcal{T}_1$ and $\mathcal{T}_2$ are closed, $\mathcal{T} = \mathcal{T}_1 \cap \mathcal{T}_2$, the intersection model's tangent space is the intersection of tangent spaces. For detail, see the proof of Theorem~\ref{thm:general_ci_model_orthocomp} in the Appendix. This implies (\citep{VanDerLaan.Robins.2003.UnifiedMethods}, Lemma 1.7)

\begin{lemma}
    \label{lm:bell_model_orthocomp}
    The orthogonal complement of the tangent space of the Bell scenario model is the direct sum $\mathcal{T}^{\perp} = \mathcal{T}_1^{\perp} \oplus \mathcal{T}_2^{\perp}$, whose explicit expression is
    \begin{equation}
    \label{eq:bell_model_orthocomp}
    \begin{aligned}
        \mathcal{T}^{\perp} & = \{ \mathbb{E}[h | a, d, b] - \mathbb{E}[h | a] - \mathbb{E}[h | d, b]
        \\
        + & \mathbb{E}[g | b, c, a] - \mathbb{E}[g | b] - \mathbb{E}[g | c, a] : \forall g,h \in \mathcal{H} \}.
    \end{aligned}
    \end{equation}
\end{lemma}

\begin{proof}
Since $\mathcal{T} = \mathcal{T}_1 \cap \mathcal{T}_2$, elements of $\mathcal{T}$ must be orthogonal to both $\mathcal{T}_1^\perp$ and $\mathcal{T}_2^\perp$. This implies the span $\overline{\mathcal{T}_1^\perp \oplus \mathcal{T}_2^\perp}$ is in the orthogonal complement of $\mathcal{T}$. Since both $\mathcal{T}_1^\perp$ and $\mathcal{T}_2^\perp$ are closed, $\overline{\mathcal{T}_1^\perp \oplus \mathcal{T}_2^\perp} = \mathcal{T}_1^\perp \oplus \mathcal{T}_2^\perp \subseteq \mathcal{T}^\perp$. For other direction, pick any $f \in \left(\mathcal{T}_1^\perp \oplus \mathcal{T}_2^\perp\right)^\perp$. Then $f$ must be orthogonal to just $\mathcal{T}_1^\perp$, so $f \in \mathcal{T}_1$. Similarly, $f \in \mathcal{T}_2$, and therefore $f \in \mathcal{T}$. We just show $\left(\mathcal{T}_1^\perp \oplus \mathcal{T}_2^\perp\right)^\perp \subseteq \mathcal{T}$, and by taking the orthogonal complement we get $\mathcal{T}^\perp \subseteq \mathcal{T}_1^\perp \oplus \mathcal{T}_2^\perp$. Therefore, $\mathcal{T}^\perp = \mathcal{T}_1^\perp \oplus \mathcal{T}_2^\perp$. Finally, the explicit expression for $\mathcal{T}^\perp$ follows from the definition of direct sum and the explicit expressions for $\mathcal{T}_1^\perp$ and $\mathcal{T}_2^\perp$. 
\end{proof}

Consequently, if $\varphi$ is a known influence function of a target parameter in this model, then
\begin{equation}
\label{eqn:bell-orth}
\begin{aligned}
    \varphi & + \mathbb{E}[h \mid a, b, d] - \mathbb{E}[h \mid a] - \mathbb{E}[h \mid b, d]
    \\
    & + \mathbb{E}[g \mid b, a, c] - \mathbb{E}[g \mid b] - \mathbb{E}[g \mid a, c]
\end{aligned}
\end{equation}
is also an influence function of that target parameter, where $g,h$ are arbitrary functions in $\mathcal{H}$.
Equation~\ref{eqn:bell-orth} thus characterizes the class of influence functions for any target parameter in the Bell scenario model.

\subsection{Main Result}
\label{sec:ci_tangent:thm}

The result for the Bell scenario model described in the previous section can easily be extended to any Markov model, including all mentioned graphical models associated to UGs, CGs, and ADMGs. Consider a general Markov model defined by a list of $K \geq 1$ marginal and conditional independence constraints, as model without any constraint ($K=0$) is the saturated model described in \citet{Tsiatis.2006.SemiparametricTheory}. By definition, this Markov model is the intersection of $K$ single independence DAG models defined by a single constraint $\mathscr{P} = \mathscr{P}_1 \cap \ldots \cap \mathscr{P}_K$. The tangent space is the intersection of the tangent spaces of these single independence models, $\mathcal{T} = \mathcal{T}_1 \cap \ldots \cap \mathcal{T}_K$, and its orthogonal complement is the direct sum $\mathcal{T}^\perp = \mathcal{T}_1^\perp \oplus \ldots \oplus \mathcal{T}_K^\perp$.
\begin{theorem}
\label{thm:general_ci_model_orthocomp}
    Suppose $\mathscr{P}$ is a general semi-parametric Markov model defined by $K \geq 1$ conditional independence constraints $\mathbf{X}_i \Perp \mathbf{Y}_i \mid \mathbf{Z}_i$, where $i = 1, \ldots, K$. For each $i$, let $\mathscr{P}_i$ be the single independence DAG model defined by the $i$-th constraint. The orthogonal projection operator $\Pi(\cdot \mid \mathcal{T}_i^\perp): \mathcal{H} \mapsto \mathcal{H}$ sends each function $h \in \mathcal{H}$ to the orthogonal complement $\mathcal{T}_i^\perp$ of the tangent space $\mathcal{T}_i$ of $\mathscr{P}_i$,
    \begin{equation}
    \label{eq:proj_operator_single_dag_model}
    \begin{aligned}
        \Pi(h \mid \mathcal{T}_i^\perp) (\mathbf{x}_{i}, \mathbf{y}_{i}, \mathbf{z}_{i})
        = \mathbb{E}[h | \mathbf{x}_{i}, \mathbf{y}_{i}, \mathbf{z}_{i}] - \mathbb{E}[h | \mathbf{x}_{i}, \mathbf{z}_{i}] &
        \\
        - \mathbb{E}[h | \mathbf{y}_{i}, \mathbf{z}_{i}]
        + \mathbb{E}[h | \mathbf{z}_{i}]&.
    \end{aligned}
    \end{equation}
    The orthogonal complement of the tangent space of $\mathscr{P}$ is the direct sum of these orthogonal complements,
    \begin{equation}
    \label{eq:dag_orthocomp_general}
    \begin{aligned}
        \mathcal{T}^{\perp} = \left\{ \sum_{i = 1}^K \Pi(h_i \mid \mathcal{T}_i^\perp) \: : \: \forall h_1, \ldots, h_K \in \mathcal{H} \right\}.
    \end{aligned}
    \end{equation}
\end{theorem}
\begin{proof}[Proof Sketch]
Applying Lemma~\ref{lm:single_dag_model_orthocomp} to the single independence model $\mathcal{P}_i$, any element of the tangent space's orthogonal complement $\mathcal{T}_i^\perp$ is $\Pi(h_i \mid \mathcal{T}_i^\perp)$, with arbitrary $h_i \in \mathcal{H}$. A general version of Lemma~\ref{lm:bell_model_orthocomp} shows that $\mathcal{T}^\perp = \mathcal{T}_1^\perp \oplus \ldots \oplus \mathcal{T}_K^\perp$. Hence, any element of $\mathcal{T}^{\perp}$ is a sum $\sum_{i = 1}^K \Pi(h_i \mid \mathcal{T}_i^\perp)$ of elements of $\{\mathcal{T}_i^\perp\}_{i=1}^K$, where $\{h_i \in \mathcal{H}\}_{i=1}^K$ are 
arbitrary (not necessarily equal) functions.
\end{proof}

This theorem applies to all Markov models -- models defined by a list of CI constraints -- which are not necessarily graphical models. Furthermore, for graphical Markov models such as those associated with UGs, DAGs, CGs, BGs, and ADMGs, the list of CI constraints can be obtained from global, local or any other Markov property, as long as this list fully describes the model. This means that it is possible to obtain several syntatically different representations of the same orthocomplement of the tangent space of the model by using different Markov properties. For example, the local Markov property for DAGs listed in Section~\ref{sec:ci_models} differs from the ordered local Markov property for DAGs, which is used in \citet{Bhattacharya.Nabi.ea.2022.SemiparametricInference, Tsiatis.2006.SemiparametricTheory, Rotnitzky.Smucler.2020.EfficientAdjustment}. Despite the different expressions, the orthogonal complement of the tangent space is the same.

Equation~\ref{eq:dag_orthocomp_general} suggests an operator $\Gamma: \mathcal{H} \mapsto \mathcal{T}^\perp$ sending functions in $\mathcal{H}$ to functions in the orthogonal complement $\mathcal{T}^\perp$, defined by $\Gamma(f) = \sum_{i = 1}^K \Pi(f \mid \mathcal{T}_i^\perp)$. For general Markov model, it is unknown whether this operator is the projection operator onto $\mathcal{T}^\perp$. The Appendix shows
an example in the Bell scenario case.
Therefore, finding the projection operator onto a general Markov model's tangent space or its orthogonal complement is still an open problem. Since the efficient influence function (EIF) is $\varphi - \Pi(\varphi \mid \mathcal{T}^\perp)$, where $\varphi$ is any influence function and $\Pi(\cdot \mid \mathcal{T}^\perp)$ is the projection onto the orthogonal complement $\mathcal{T}^\perp$ of the tangent space, the exact form of the efficient influence function for a target parameter in a general Markov model is also an open problem.

A direct consequence of Theorem~\ref{thm:general_ci_model_orthocomp} is the explicit characterization of the class of all influence functions of any target parameter $\psi(p)$ in a general Markov model. If $\varphi$ is a known influence function of $\psi(p)$, then
\begin{equation}
    \varphi + \sum_{i = 1}^K \Pi(h_i \mid \mathcal{T}_i^\perp)
\end{equation}
is also an influence function of $\psi(p)$, for any $\{h_i \in \mathcal{H}\}_{i=1}^K$. An explicit example is shown in the next section.

\section{Applications}
\label{sec:application}

\subsection{The Class of Influence Functions}
\label{sec:application:class}

In this subsection, we illustrate how our result streamline the derivations of the class of influence functions for the conditional mean parameter $\psi(p) = \mathbb{E}[D \mid A = a_0]$ with a fixed value $a_0$, for all the Markov models corresponding to the graphs in Figure~\ref{fig:ci_models}. Extensions to other target parameters are trivial. For example, if the target parameter is the adjustment formula $\mathbb{E}[\mathbb{E}[D \mid A = a_0, C]]$, simply replace $\varphi$ below by the augmented inverse probability weighting (AIPW) functional.

\paragraph{The saturated model $\mathscr{P}^{\text{all}}_{ABCD}$ in Figure~\ref{fig:ci_models}f.} The saturated model $\mathscr{P}^{\text{all}}_{ABCD}$ over four variables $\{A,B,C,D\}$ is a saturated extension over the saturated model $\mathscr{P}^{\text{all}}_{AD}$ with only two variables $\{A,D\}$, with $B$ and $C$ play the role of auxillary variables. Following \citet{Tsiatis.2006.SemiparametricTheory}, Chapter 5.5, the set of influence functions in $\mathscr{P}^{\text{all}}_{ABCD}$ is exactly the set of influence functions in $\mathscr{P}^{\text{all}}_{AD}$. This means, for the purpose of estimating $\mathbb{E}[D \mid A = a_0]$, we can safely ignore $B$ and $C$ as these variables do not provide efficiency gain. Moreover, the derivation of the influence functions in $\mathscr{P}^{\text{all}}_{AD}$ involves only two variables and hence is much simpler.

Any saturated model has $\{0\}$ as its orthogonal complement of the tangent space, so there is only one influence function,
\begin{equation}
\label{eq:if_complete_dag}
    \varphi(a,b,c,d) = \frac{\mathbb{I}(a = a_0)}{p(a)}(d - \mathbb{E}[D | a]).
\end{equation}
All other Markov models corresponding to the graphs in Figure~\ref{fig:ci_models}a-e are submodels of $\mathscr{P}^{\text{all}}_{ABCD}$. Therefore, $\varphi$ is also an influence function in any of these models. The other ingredient for the derivation the class of influence functions in these models is the orthogonal complement of the tangent space given by Theorem~\ref{thm:general_ci_model_orthocomp}.

\paragraph{The DAG model $\mathscr{P}^{(a)}$ in Figure~\ref{fig:ci_models}a.}
Let $\mathscr{P}_1$ and $\mathscr{P}_2$ be the simple models corresponding to the constraints $C \Perp B \mid A$ and $D \Perp A \mid C,B$, respectively. Let $\mathcal{T}_1^\perp$ and $\mathcal{T}_2^\perp$ denote the orthocomplements of their respective tangent spaces. The projection operators onto these orthocomplements are
\begin{equation*}
\begin{aligned}
    \Pi_a(g \mid \mathcal{T}_1^\perp) = & \mathbb{E}[g | c,b,a] - \mathbb{E}[g | c,a] - \mathbb{E}[g | b,a] + \mathbb{E}[g | a],
    \\
    \Pi_a(h \mid \mathcal{T}_2^\perp) = & h - \mathbb{E}[h | d,c,b] - \mathbb{E}[h | a,c,b] + \mathbb{E}[h | c,b].
\end{aligned}
\end{equation*}
Therefore, the orthogonal complement of the tangent space $\mathcal{T}$ for this DAG model is, according to Equation~\ref{eq:dag_orthocomp_general}
\begin{equation}
\begin{aligned}
\label{eq:dag_orthocomp_2}
    \mathcal{T}^\perp = \{\Pi_a (g \mid \mathcal{T}_1^\perp) + \Pi_a (h \mid \mathcal{T}_2^\perp): \forall g,h \in \mathcal{H}\},
\end{aligned}
\end{equation}
This expression differs from the one given in Equation~\ref{eq:dag_orthocomp}. To reconcile the two, note that $\mathcal{T}^\perp$ is the direct sum of the orthogonal subspaces $\mathcal{T}_1^\perp$ and $\mathcal{T}_2^\perp$. This geometric property holds for general DAGs: the orthocomplements of the tangent spaces corresponding to distinct CIs in the local Markov property are pairwise orthogonal. Since the projection operators onto these subspaces are known -- namely, $\Pi_a (\cdot \mid \mathcal{T}_1^\perp)$ and $\Pi_a (\cdot \mid \mathcal{T}_2^\perp)$ -- the projection of a function $h \in \mathcal{H}$ onto $\mathcal{T}^\perp$ is given by $\Pi(h \mid \mathcal{T}^\perp) = \Pi_a (h \mid \mathcal{T}_1^\perp) + \Pi_a (h \mid \mathcal{T}_2^\perp)$. This shows that the two expressions for $\mathcal{T}^\perp$ in \ref{eq:dag_orthocomp} and \ref{eq:dag_orthocomp_2} describe the same subspace.

The class of influence functions for $\psi(p)$ in this model is $\{\varphi + \Pi_a (g \mid \mathcal{T}_1^\perp) + \Pi_a (h \mid \mathcal{T}_2^\perp): \forall g, h \in \mathcal{H}\}$.

\paragraph{The UG model $\mathscr{P}^{(b)}$ in Figure~\ref{fig:ci_models}b.}
$\mathscr{P}_1$ and $\mathscr{P}_2$ correspond to $C \Perp B \mid A, D$ and $D \Perp A \mid C, B$, respectively.
\begin{equation*}
\begin{aligned}
    \Pi_b (g \mid \mathcal{T}_1^\perp)
    = & g - \mathbb{E}[g \mid c,a,d] - \mathbb{E}[g | b,a,d] + \mathbb{E}[g | a, d],
    \\
    \Pi_b (h \mid \mathcal{T}_2^\perp)
    = & h - \mathbb{E}[h | d,c,b] - \mathbb{E}[h | a,c,b] + \mathbb{E}[h | c,b].
\end{aligned}
\end{equation*}
The class of influence functions for $\psi(p)$ in this model is $\{\varphi + \Pi_b (g \mid \mathcal{T}_1^\perp) + \Pi_b (h \mid \mathcal{T}_2^\perp): \forall g, h \in \mathcal{H}\}$.

\paragraph{The CG model $\mathscr{P}^{(c)}$ in Figure~\ref{fig:ci_models}c.}
$\mathscr{P}_1$ and $\mathscr{P}_2$ are the same models in the UG case, while $\mathscr{P}_3$ corresponds to $A \Perp B$.
The projection operator onto $\mathcal{T}_3^\perp$ is
\begin{equation*}
\begin{aligned}
    \Pi_c (f \mid \mathcal{T}_3^\perp) = \mathbb{E}[f \mid a,b] - \mathbb{E}[f \mid a] - \mathbb{E}[f \mid b].
\end{aligned}
\end{equation*}
The class of influence functions for $\psi(p)$ in this model is $\{\varphi + \Pi_b (g | \mathcal{T}_1^\perp) + \Pi_b (h | \mathcal{T}_2^\perp) + \Pi_c (f | \mathcal{T}_3^\perp): \forall f, g, h \in \mathcal{H}\}$.

\paragraph{The bidirected square model $\mathscr{P}^{(e)}$ in Figure~\ref{fig:ci_models}e.}
$\mathscr{P}_1$ and $\mathscr{P}_2$ correspond to $A \Perp D$ and $B \Perp C$, respectively.
\begin{equation*}
\begin{aligned}
    \Pi_d (g \mid \mathcal{T}_1^\perp) & = \mathbb{E}[g \mid a,d] - \mathbb{E}[g \mid a] - \mathbb{E}[g \mid d],
    \\
    \Pi_d (h \mid \mathcal{T}_2^\perp) & = \mathbb{E}[h \mid b,c] - \mathbb{E}[h \mid b] - \mathbb{E}[h \mid c].
\end{aligned}
\end{equation*}
The class of influence functions for $\psi(p)$ in this model is
$\{\varphi + \Pi_d (g \mid \mathcal{T}_1^\perp) + \Pi_d (h \mid \mathcal{T}_2^\perp): \forall g, h \in \mathcal{H}\}$.

\subsection{Improving Efficiency}
\label{sec:application:improve}

The variance of a RAL estimator $\hat{\psi}$ for a target parameter $\psi(p)$ is exactly the variance of the corresponding influence function $\varphi \in \mathcal{H}$. For influence functions have mean zero, the variance of $\varphi$ is its squared norm $\mathbb{E}[\varphi^2]$. Thus, an estimator $\hat{\psi}_2$ is at least as efficient as another estimator $\hat{\psi}_1$ precisely when the variance of the influence function $\varphi_2$ of $\hat{\psi}_2$ is equal or smaller than the variance of the influence functions $\varphi_1$ of $\hat{\psi}_1$, $\mathbb{E}[\varphi_2^2] \leq \mathbb{E}[\varphi_1^2]$ (\citet{Tsiatis.2006.SemiparametricTheory}, Chapter 3). The efficient influence function is simply the influence function with the smallest variance in the class of all influence functions of the model. Therefore, in order to construct more efficient estimators, the first step is finding influence functions with smaller variances.

Given an initial influence function $\varphi_0$ of a target parameter, our main result immediately implies an iterative method for constructing sequences of influence functions with non-increasing variances. As an example, suppose the Markov model has $K \geq 2$ CI constraints. Theorem~\ref{thm:general_ci_model_orthocomp} states that any $\sum_{i=1}^K \Pi(h_i \mid \mathcal{T}_i^\perp)$ is an element of $\mathcal{T}^\perp$. Therefore, $\Pi(\varphi_0 \mid \mathcal{T}_1^\perp) \in \mathcal{T}^\perp$. By Theorem 4.3 of \cite{Tsiatis.2006.SemiparametricTheory}, adding an element in $\mathcal{T}^\perp$ to an influence function yields another influence function of the same target, hence
\begin{equation}
    \varphi_1 := \varphi_0 - \Pi(\varphi_0 \mid \mathcal{T}_1^\perp)
\end{equation}
is an influence function. Moreover, since $\Pi(\varphi_0 \mid \mathcal{T}_1^\perp)$ is the orthogonal projection of $\varphi_0$ onto $\mathcal{T}_1^\perp$, the above $\varphi_1$ is the orthogonal projection of $\varphi_0$ onto $\mathcal{T}_1$. The Pythagorean theorem for Hilbert spaces (\citet{Tsiatis.2006.SemiparametricTheory}, Theorem 3.3) states that the projection of a vector cannot be longer than the vector itself, so $E[\varphi_1^2] \leq E[\varphi_0^2]$. This implies that $\varphi_1$ is an influence function with equal or smaller variance than $\varphi_0$. Next, consider
\begin{equation}
    \varphi_2 := \varphi_1 - \Pi(\varphi_1 \mid \mathcal{T}_2^\perp).
\end{equation}
By the same reasoning, $\varphi_2$ is an influence function with equal or smaller variance than $\varphi_1$. Repeating this process $M$ times, we obtain a sequence of influence functions $\varphi_0, \varphi_1, \varphi_2, \ldots, \varphi_{M}$, with the property that the variance of $\varphi_{m}$ is equal or smaller than the variance of $\varphi_{m-1}$, for all $m=1, \ldots, M$. The influence function $\varphi_m$ is referred to as an $m$-improvement of the initial $\varphi_0$. If we can construct an estimator $\hat{\psi}_{M}$ based on the influence function $\varphi_{M}$, then this estimator is at least as efficient as the estimator $\hat{\psi}_0$ based on the initial influence function $\varphi_0$.
In the Appendix, we give a concrete example of such a construction with $M=2$.

\begin{theorem}
\label{thm:efficiency_improvement}
    Consider a Markov model with $K \geq 1$ constraint(s), and let $\varphi_0$ be an influence function of a target parameter $\psi(p)$. For any integer sequence $i_1, i_2, \ldots, i_{M}$, where $i_m \in \{1, \ldots, K\}$ for all $m \in \{1, \ldots, M\}$, define
    \begin{equation}
    \begin{aligned}
        \varphi_m := \varphi_{m-1} - \Pi(\varphi_{m-1} \mid \mathcal{T}_{i_m}^\perp), \quad \forall m \in \{1, \ldots, M\}.
    \end{aligned}
    \end{equation}
    In words, $\varphi_m$ is the projection of $\varphi_{m-1}$ onto the tangent space $\mathcal{T}_{i_m}$ of the DAG model corresponding to the $i_m$-th constraint. Then $\varphi_0, \ldots, \varphi_{M}$ is a sequence of influence functions of $\psi(p)$ with non-increasing variance, meaning $E[\varphi_{m}^2] \leq E[\varphi_{m-1}^2]$ for all $m \in \{1, \ldots, M\}$.
\end{theorem}

If $i_{m+1} = i_m$, then $\varphi_{m+1} = \varphi_m$ is the same influence function, because $\varphi_m \in \mathcal{T}_{i_m}$ and so $\Pi(\varphi_{m} \mid \mathcal{T}_{i_{m+1}}^\perp) = 0$. Then, for single constraint models ($K=1$), there is only one improvement, $\varphi = \varphi_{0} - \Pi(\varphi_{0} \mid \mathcal{T}^\perp)$, which is exactly the efficient influence function in these models. In models with $K > 1$, the connection between the efficient influence function and $\varphi_M$ as $M \rightarrow \infty$ is an interesting question, yet outside the scope of our paper.


\section{
Conclusion}
\label{sec:discussion}

We derived an explicit expression for the orthogonal complement of the tangent space for Markov models -- statistical models defined using only marginal and conditional independence constraints.

We demonstrated two applications of our result in these models.
The first one is an explicit characterization of the class of influence functions for any target parameter. The second one is is an iterative method to improve efficiency of any initial RAL estimator of any target parameter.

Our method does not yield an explicit expression for the likelihood or the tangent space of these models, so these questions remain open. Another important open problem is the projection of a function $h \in \mathcal{H}$ onto 
the tangent space of these models. Derivation of such a projection is equivalent to the derivation of the efficient influence function in a semi-parametric Markov model.
Our method provides an important step towards achieving
this goal by providing the explicit form of the orthogonal complement of the tangent space.






\begin{acknowledgements} 
    This work was supported by the Office of Naval Research
(Grant No. N000142412701), the National Institutes of
Health (Grant No. 1R56AI191526-01 and Grant no. 1R01LM014800-01A1), and the National
Science Foundation CAREER Award (Grant No. 1942239).
\end{acknowledgements}

\bibliography{uai2026-template}

\newpage

\onecolumn

\title{A Characterization of the Orthocomplement of the Tangent Space of Semiparametric Markov Models\\(Supplementary Material)}
\maketitle

\appendix


\section{Additional Details about Semiparametric Theory}
\label{sec:additional_semi_param}

\paragraph{The set of all influence function and the orthogonal complement of the tangent space.}

As mentioned in the main paper, an influence function of the target parameter $\psi(p)$ is a solution of the pathwise differentiability condition, for all parametric submodel $p_{\varepsilon}$ (\citet{Tsiatis.2006.SemiparametricTheory}, Theorem 3.2 and 4.2)
\begin{equation*}
    \frac{\partial}{\partial \varepsilon} \psi(p_{\varepsilon}) \bigg|_{\varepsilon=0} = \int \varphi(\mathbf{v}) s(\mathbf{v}) p_0(\mathbf{v}) d \mathbf{v}.
    \quad \quad \quad \eqref{eq:path_dev}
\end{equation*}
Suppose $\varphi$ is an influence function of $\psi(p)$ and $\mathcal{T}^\perp$ is the orthogonal complement of the tangent space of the model. Let $f$ be any function in $\mathcal{T}^\perp$. Since the score function $s(\mathbf{v}) \in \mathcal{T}$, it is orthogonal to $f$, that is $\langle f, s \rangle = \int f(\mathbf{v}) s(\mathbf{v}) p_0(\mathbf{v}) d \mathbf{v} = 0$. Then $\varphi + f$ is also an influence function, because
\begin{equation}
    \frac{\partial}{\partial \varepsilon} \psi(p_{\varepsilon}) \bigg|_{\varepsilon=0} = \int \big[\varphi(\mathbf{v}) + f(\mathbf{v})\big] s(\mathbf{v}) p_0(\mathbf{v}) d \mathbf{v}.
\end{equation}
On the other hand, if $\tilde{\varphi}$ is any other influence function, meaning
\begin{equation}
    \frac{\partial}{\partial \varepsilon} \psi(p_{\varepsilon}) \bigg|_{\varepsilon=0} = \int \tilde{\varphi}(\mathbf{v}) s(\mathbf{v}) p_0(\mathbf{v}) d \mathbf{v}.
\end{equation}
Then, for all parametric submodel $p_{\varepsilon}$ (equivalently, for all score function $s(\mathbf{v}) \in \mathcal{T}$)
\begin{equation}
    0 = \int \big[\tilde{\varphi}(\mathbf{v}) - \varphi(\mathbf{v})\big] s(\mathbf{v}) p_0(\mathbf{v}) d \mathbf{v}.
\end{equation}
This means $\tilde{\varphi} - \varphi \in \mathcal{T}^\perp$, i.e., $\tilde{\varphi}$ is the sum of $\varphi$ and an element of the orthogonal complement of the tangent space. This result is Theorem 4.3 in \citet{Tsiatis.2006.SemiparametricTheory}, stating that: the set of all influence functions of a target parameter in a semiparametric model is the linear variety $\varphi \oplus \mathcal{T}^\perp := \{\varphi + f : f \in \mathcal{T}^\perp\}$, where $\varphi$ is any influence function and $\mathcal{T}^\perp$ is the orthogonal complement of the tangent space. Hence, knowing the orthogonal complement of the tangent space $\mathcal{T}^\perp$ is vital to characterize the set of all influence functions.

\paragraph{Variance of an influence function and efficiency of an estimator.} As stated in the main paper, the limiting distribution of a RAL estimator $\hat{\psi}_n$ of the target parameter $\psi$ is
\begin{equation*}
    \sqrt{n} (\hat{\psi}_n - \psi) \rightsquigarrow \mathcal{N}(0, \mathbb{E}[\varphi^2])
    \quad \quad \quad \eqref{eq:asymp_linear_clt}
\end{equation*}
where $\varphi$ is the influence function corresponding to $\hat{\psi}$ \citep{Tsiatis.2006.SemiparametricTheory}. Since $\varphi \in \mathcal{H}$, it has zero mean, and therefore the variance of the influence function $\varphi$ is its square norm in $\mathcal{H}$, $\operatorname{Var}(\varphi) = E[\varphi^2] - E[\varphi]^2 = E[\varphi^2] = \langle \varphi, \varphi \rangle$.

An estimator $\hat{\psi}$ is more efficient than another estimator $\hat{\psi}'$ precisely when the variance of the limiting distribution of $\sqrt{n} (\hat{\psi}_n - \psi)$ is equal or smaller than the variance of the limiting distribution of $\sqrt{n} (\hat{\psi}'_n - \psi)$ \citep{Tsiatis.2006.SemiparametricTheory}. This means the variance of the influence function $\varphi$ of $\hat{\psi}$ is equal or smaller than the variance of the influence function $\varphi'$ of $\hat{\psi}'$. Therefore, given an influence function $\varphi$ and the orthogonal complement of the tangent space $\mathcal{T}^\perp$, if one can construct $\varphi - f$, where $f \in \mathcal{T}^\perp$, such that $\EE[(\varphi - f)^2] \leq \EE[\varphi^2]$, then the estimator obtained from $\varphi - f$ will be equal or more efficient than the estimator obtained from $\varphi$. In particular, we derive the orthogonal complement of the tangent space for general Markov model in Theorem~\ref{thm:general_ci_model_orthocomp}, and propose a method to improve the variance of an initial influence function in Theorem~\ref{thm:efficiency_improvement}.

The efficient influence function is the influence function $\varphi^{(e)}$ in the class of all influence functions $\varphi \oplus \mathcal{T}^\perp$ with smallest variance, i.e., $\EE[(\varphi^{(e)})^2] \leq \EE [\varphi'^2]$ for any influence function $\varphi' \in \varphi \oplus \mathcal{T}^\perp$ \citep{Tsiatis.2006.SemiparametricTheory}. One way to find the efficient influence function is projecting any other influence function $\varphi$ onto the tangent space $\mathcal{T}$. Let $\Pi(\cdot \mid \mathcal{T})$ be the projection onto $\mathcal{T}$, and $\Pi(\cdot \mid \mathcal{T}^\perp)$ is the projection onto the orthogonal complement $\mathcal{T}^\perp$. We have $\Pi(\cdot \mid \mathcal{T}^\perp) = I - \Pi(\cdot \mid \mathcal{T})$. For any influence function $\varphi$
\begin{equation}
    \varphi = \varphi - \Pi(\varphi \mid \mathcal{T}^\perp) + \Pi(\varphi \mid \mathcal{T}^\perp)
\end{equation}
Since $\Pi(\varphi \mid \mathcal{T}^\perp)$ is an element in the orthogonal complement of the tangent space, $\varphi - \Pi(\varphi \mid \mathcal{T}^\perp)$ must be another influence function. Since $\varphi - \Pi(\varphi \mid \mathcal{T}^\perp) = \Pi(\varphi \mid \mathcal{T})$ is the projection of $\varphi$ onto a subspace, its norm must be equal or less than the norm of $\varphi$, so $\EE [(\varphi - \Pi(\varphi \mid \mathcal{T}^\perp))^2] \leq \EE [\varphi^2]$. It can be shown (\citet{Tsiatis.2006.SemiparametricTheory}, Theorem 3.5) that the efficient influence function is $\varphi^{(e)} = (\varphi - \Pi(\varphi \mid \mathcal{T}^\perp)$. Hence, to find the efficient influence function, one must first know the projection operator $\Pi(\cdot \mid \mathcal{T}^\perp)$ or $\Pi(\cdot \mid \mathcal{T})$.

\section{Proofs in Section \ref{sec:dag_tangent}}

The following lemma is the general version of Equation~\ref{eq:dag_subtangent}.
\begin{lemma}[\citet{Tsiatis.2006.SemiparametricTheory}, Theorem 4.5]
    \label{lm:tangent_fact_1}
    Let $\mathscr{P}$ be a semi-parametric model over variables $\mathbf{V}$. Consider any disjoint subsets $\mathbf{X}, \mathbf{Y} \subseteq \mathbf{V}$. For any parametric submodel $p_\varepsilon$, the score function $s(\mathbf{x} \mid \mathbf{y}) = \frac{\partial}{\partial \varepsilon} \log p_\varepsilon (\mathbf{x} \mid \mathbf{y}) |_{\varepsilon=0}$ at $p_0$ is an element of the subspace
    \begin{equation}
        \mathcal{T}_{\mathbf{X} \mid \mathbf{Y}} := \{E[h \mid \mathbf{x}, \mathbf{y}] - E[h \mid \mathbf{y}] : \forall h \in \mathcal{H}\}.
    \end{equation}
\end{lemma}

\begin{proof}
Let $\mathbf{Z} = \mathbf{V} \setminus (\mathbf{X} \dot{\cup} \mathbf{Y})$. The following is a really useful identify for the score function
\begin{equation}
\begin{aligned}
    s(\mathbf{x}, \mathbf{y}, \mathbf{z})
    & := \frac{\partial}{\partial \varepsilon} \log p_\varepsilon (\mathbf{x}, \mathbf{y}, \mathbf{z}) \bigg|_{\varepsilon=0}
    = \frac{1}{p_0 (\mathbf{x}, \mathbf{y}, \mathbf{z})} \frac{\partial}{\partial \varepsilon} p_\varepsilon (\mathbf{x}, \mathbf{y}, \mathbf{z}) \bigg|_{\varepsilon=0}.
\end{aligned}
\end{equation}
Then
\begin{equation}
\begin{aligned}
    s(\mathbf{x} \mid \mathbf{y})
    & := \frac{\partial}{\partial \varepsilon} \log p_\varepsilon (\mathbf{x} \mid \mathbf{y}) \bigg|_{\varepsilon=0}
    = \frac{\partial}{\partial \varepsilon} \log p_\varepsilon (\mathbf{x}, \mathbf{y}) \bigg|_{\varepsilon=0} - \frac{\partial}{\partial \varepsilon} \log p_\varepsilon (\mathbf{y}) \bigg|_{\varepsilon=0}
    \\
    & = \frac{1}{p_0 (\mathbf{x}, \mathbf{y})} \frac{\partial}{\partial \varepsilon} \int p_\varepsilon (\mathbf{x}, \mathbf{y}, \mathbf{z}) d \mathbf{z} \bigg|_{\varepsilon=0}  - \frac{1}{p_0 (\mathbf{y})} \frac{\partial}{\partial \varepsilon} \int p_\varepsilon (\mathbf{x}, \mathbf{y}, \mathbf{z}) d \mathbf{x} d \mathbf{z} \bigg |_{\varepsilon=0}
    \\
    & = \frac{1}{p_0 (\mathbf{x}, \mathbf{y})} \int \left(\frac{\partial}{\partial \varepsilon} p_\varepsilon (\mathbf{x}, \mathbf{y}, \mathbf{z}) \bigg|_{\varepsilon=0} \right) d \mathbf{z}  - \frac{1}{p_0 (\mathbf{y})} \int \left(\frac{\partial}{\partial \varepsilon} p_\varepsilon (\mathbf{x}, \mathbf{y}, \mathbf{z})\bigg |_{\varepsilon=0}\right) d \mathbf{x} d \mathbf{z}
    \\
    & = \frac{1}{p_0 (\mathbf{x}, \mathbf{y})} \int s(\mathbf{x}, \mathbf{y}, \mathbf{z}) p_0 (\mathbf{x}, \mathbf{y}, \mathbf{z}) d \mathbf{z}  - \frac{1}{p_0 (\mathbf{y})} \int s(\mathbf{x}, \mathbf{y}, \mathbf{z}) p_0 (\mathbf{x}, \mathbf{y}, \mathbf{z}) d \mathbf{x} d \mathbf{z}
    \\
    & = \int s(\mathbf{x}, \mathbf{y}, \mathbf{z}) p_0 (\mathbf{z} \mid \mathbf{x}, \mathbf{y}) d \mathbf{z}  - \int s(\mathbf{x}, \mathbf{y}, \mathbf{z}) p_0 (\mathbf{x}, \mathbf{z} \mid \mathbf{y}) d \mathbf{x} d \mathbf{z}
    \\
    & = E[s(\mathbf{V}) \mid \mathbf{x}, \mathbf{y}] - E[s(\mathbf{V}) \mid \mathbf{y}].
\end{aligned}
\end{equation}
Since $s(\mathbf{v}) \in \mathcal{H}$, the claim is established.
\end{proof}

The following lemma shows that the tangent subspaces of the two pathway DAG model (Figure~\ref{fig:ci_models}a) described in Equation~\ref{eq:dag_subtangent} are orthogonal.

\begin{lemma}[\citet{Tsiatis.2006.SemiparametricTheory}, Theorem 4.5]
    \label{lm:tangent_fact_2}
    Let $\mathbf{A}, \mathbf{B}, \mathbf{C}, \mathbf{D}$ be disjoint sets of random variables. The following 2 subspaces are orthogonal
    \begin{equation}
    \begin{aligned}
        \mathcal{T}_1 & = \{ \EE [h \mid \mathbf{a}, \mathbf{b}, \mathbf{c}, \mathbf{d}] - \EE [h \mid \mathbf{b}, \mathbf{c}, \mathbf{d}]: \forall h \in \mathcal{H}\},
        \\
        \mathcal{T}_2 & = \{ \EE [h \mid \mathbf{c}, \mathbf{d}] - \EE [h \mid \mathbf{d}]: \forall h \in \mathcal{H}\}.
    \end{aligned}
    \end{equation}
\end{lemma}
\begin{proof}
    Pick any $\EE [h_1 \mid \mathbf{a}, \mathbf{b}, \mathbf{c}, \mathbf{d}] - \EE [h_1 \mid \mathbf{b}, \mathbf{c}, \mathbf{d}] \in \mathcal{T}_1$ and any $\EE [h_2 \mid \mathbf{c}, \mathbf{d}] - \EE [h \mid \mathbf{d}] \in \mathcal{T}_2$, then consider their inner product
    \begin{equation}
    \begin{aligned}
        & \EE \bigg[ \left(\EE [h_1 \mid \mathbf{A}, \mathbf{B}, \mathbf{C}, \mathbf{D}] - \EE [h_1 \mid \mathbf{B}, \mathbf{C}, \mathbf{D}] \right) \left( \EE [h_2 \mid \mathbf{C}, \mathbf{D}] - \EE [h_2 \mid \mathbf{D}] \right) \bigg]
        \\
        & = \int \bigg( \int \left(\EE [h_1 \mid \mathbf{a}, \mathbf{b}, \mathbf{c}, \mathbf{d}] - \EE [h_1 \mid \mathbf{b}, \mathbf{c}, \mathbf{d}] \right) \left( \EE [h_2 \mid \mathbf{c}, \mathbf{d}] - \EE [h_2 \mid \mathbf{d}] \right) p(\mathbf{a}, \mathbf{b} \mid \mathbf{c}, \mathbf{d}) d \mathbf{a} d \mathbf{b} \bigg) p(\mathbf{c}, \mathbf{d}) d \mathbf{c} d \mathbf{d}
        \\
        & = 0
    \end{aligned}
    \end{equation}
    by the law of iterated expectations.
\end{proof}

\begin{lemma}
    \label{lm:tangent_fact_3}
    The operator $\Pi: \mathcal{H} \mapsto \mathcal{H}$ defined by
    \begin{equation}
        \Pi(h) = \EE [h \mid \mathbf{x}, \mathbf{y}] - \EE [h \mid \mathbf{y}]
    \end{equation}
    is the orthogonal projection onto the subspace $\mathcal{T}_{\mathbf{X} \mid \mathbf{Y}} := \{E[h \mid \mathbf{x}, \mathbf{y}] - E[h \mid \mathbf{y}] : \forall h \in \mathcal{H}\}$.
\end{lemma}
\begin{proof}
    By definition, $\mathcal{T}_{\mathbf{X} \mid \mathbf{Y}}  = \{\Pi(h) : \forall h \in \mathcal{H}\}$. For any $f, g \in \mathcal{H}$
    \begin{equation}
        \begin{aligned}
            \langle \Pi(f), (I - \Pi)(g) \rangle
            & = \EE \bigg[ \big( \EE [f \mid \mathbf{X}, \mathbf{Y}] - \EE [f \mid \mathbf{Y}] \big) \big(g - \EE [g \mid \mathbf{X}, \mathbf{Y}] + \EE [g \mid \mathbf{Y}] \big) \bigg]
            \\
            & = \EE \bigg[ \big( \EE [f \mid \mathbf{X}, \mathbf{Y}] - \EE [f \mid \mathbf{Y}] \big) \big(\EE [g \mid \mathbf{X}, \mathbf{Y}] - \EE [g \mid \mathbf{X}, \mathbf{Y}] + \EE [g \mid \mathbf{Y}] \big) \bigg]
            \\
            & = \EE \bigg[ \big( \EE [f \mid \mathbf{X}, \mathbf{Y}] - \EE [f \mid \mathbf{Y}] \big) \EE [g \mid \mathbf{Y}] \bigg]
            \\
            & = 0,
        \end{aligned}
    \end{equation}
    by the law of iterated expectations. Any function $h \in \mathcal{H}$ can be decomposed into the sum $h = \big(h - \Pi(h) \big) + \Pi(h)$ of two orthogonal terms, where $\Pi(h) \in \mathcal{T}_{\mathbf{X} \mid \mathbf{Y}}$ and the error term $h - \Pi(h)$ is orthogonal to $\mathcal{T}_{\mathbf{X} \mid \mathbf{Y}}$. Therefore, $\Pi$ is the orthogonal projection operator onto $\mathcal{T}_{\mathbf{X} \mid \mathbf{Y}}$ \citep{Luenberger.1997.OptimizationVector}. 
\end{proof}

\begin{lemma}
    The tangent space of the two pathway DAG model in Figure~\ref{fig:ci_models}a is the direct sum
    \begin{equation*}
    \begin{aligned}
        & \mathcal{T} = \mathcal{T}_A \oplus \mathcal{T}_{B|A} \oplus \mathcal{T}_{C|A} \oplus \mathcal{T}_{D|BC}
        &
        \eqref{eq:dag_tangent} &
    \end{aligned}
    \end{equation*}
\end{lemma}

\begin{proof}
Equation~\ref{eq:dag_subtangent} shows that $\mathcal{T} \subseteq \mathcal{T}_A \oplus \mathcal{T}_{B|A} \oplus \mathcal{T}_{C|A} \oplus \mathcal{T}_{D|BC}$. To prove the equality, all we need to prove is that each subspace on the r.h.s. is contained in $\mathcal{T}$. In each case, we pick arbitrary $f$ in the subspace under consideration and define a parametric model $\mathscr{P}_{\varepsilon}$ whose elements are $p_{\varepsilon} = p_0(1 + \varepsilon f)$. We ensure $p_{\varepsilon}$ is a valid density using the conventional technique in \citet{Tsiatis.2006.SemiparametricTheory}: (i) non-negativity $p_{\varepsilon} (a,b,c,d) \geq 0$ is achieved by restricting the range of $\varepsilon$ to be $(-m, m)$ with $m$ small enough, and (ii) summation $\int p_{\varepsilon} (a,b,c,d) da db dc dd = 1 + \varepsilon E[f] = 1$ for all $\varepsilon$, because any $f \in \mathcal{H}$ must have zero-mean, by definition of $\mathcal{H}$. This parametric model contain the truth $p_0$ at $\varepsilon=0$. Moreover, the score function at $p_0$ is $\frac{\partial}{\partial \varepsilon} \log p_{\varepsilon}|_{\varepsilon = 0} = f$. Finally, if we can show $p_{\varepsilon}$ satisfies the DAG's constraints, then $\mathscr{P}_{\varepsilon}$ is a parametric submodel of the DAG's semiparametric model, and hence the chosen $f$ must be in the tangent space $\mathcal{T}$.

We will directly check $p_{\varepsilon}(c \mid b,a) = p_{\varepsilon}(c \mid a)$ and $p_{\varepsilon}(d \mid c,b,a) = p_{\varepsilon}(d \mid c,b)$ to show that $p_{\varepsilon}$ satisfies the constraints of the DAG model. This will conclude that any function in that subspace is a valid score function of this DAG model. 
Note that at the truth $p_0$, the constraint $D \Perp A \mid C, B$ implies $p_0(d \mid c,b,a) = p_0(d \mid c,b)$, and the constraint $C \Perp B \mid A$ implies $p_0(c \mid b,a) = p_0(c \mid a)$ and $p_0(b \mid c,a) = p_0(b \mid a)$. We will use these equalities repeatedly in the proof. Finally, expectations are taken at the truth $p_0$.

\textbf{First case $\mathcal{T}_A \subseteq \mathcal{T}$:} Pick any $f = \mathbb{E}[h \mid a] \in \mathcal{T}_A$. Then
\begin{equation}
\begin{aligned}
    p_{\varepsilon}(c,b,a)
    & = \int p_{0}(d,c,b,a) (1 + \varepsilon \mathbb{E}[h \mid a]) dd = p_{0}(c,b,a) (1 + \varepsilon \mathbb{E}[h \mid a])
\end{aligned}
\end{equation}
Therefore, by definition of conditional distribution
\begin{equation}
\begin{aligned}
    p_{\varepsilon}(c \mid b,a)
    & = \frac{p_{\varepsilon}(c,b,a)}{p_{\varepsilon}(b,a)} = p_{0}(c \mid b,a)
    \\
    p_{\varepsilon}(c \mid a)
    & = \frac{p_{\varepsilon}(c,a)}{p_{\varepsilon}(a)} = p_{0}(c \mid a)
    \\
    p_{\varepsilon}(d \mid c,b,a)
    & = \frac{p_{\varepsilon}(d,c,b,a)}{p_{\varepsilon}(c,b,a)} = p_{0}(d \mid c,b,a)
\end{aligned}
\end{equation}
This shows $p_{\varepsilon}(c \mid b,a) = p_{0}(c \mid b,a) = p_{0}(c \mid a) = p_{\varepsilon}(c \mid a)$. Next
\begin{equation}
\begin{aligned}
    p_{\varepsilon}(d,c,b)
    & = \int p_{0}(d,c,b) p_0(a \mid c,b) (1 + \varepsilon \mathbb{E}[h \mid a]) da
    & (A \Perp D \mid C, B)
    \\
    & = p_{0}(d,c,b) \left(1 + \varepsilon \int \mathbb{E}[h \mid a] p_0(a \mid c,b) da \right)
    \\
    p_{\varepsilon}(c,b)
    & = p_{0}(c,b) \left(1 + \varepsilon \int \mathbb{E}[h \mid a] p_0(a \mid c,b) da \right)
    \\
    \Rightarrow p_{\varepsilon}(d \mid c,b)
    & = \frac{p_{\varepsilon}(d,c,b)}{p_{\varepsilon}(c,b)} = p_{0}(d \mid c,b)
\end{aligned}
\end{equation}
This shows $p_{\varepsilon}(d \mid c,b,a) = p_{0}(d \mid c,b,a) = p_{0}(d \mid c,b) = p_{\varepsilon}(d \mid c,b)$. Hence $\mathcal{T}_A \subseteq \mathcal{T}$.

\textbf{Second case $\mathcal{T}_{B \mid A} \subseteq \mathcal{T}$:} Pick any $f = \mathbb{E}[h \mid b,a] - \mathbb{E}[h \mid a]$. Then
\begin{equation}
\begin{aligned}
    p_{\varepsilon}(c,b,a)
    & = p_{0}(c,b,a) (1 + \varepsilon \mathbb{E}[h \mid b, a] - \varepsilon \mathbb{E}[h \mid a])
    \\
    p_{\varepsilon}(b,a)
    & = p_{0}(b,a) (1 + \varepsilon \mathbb{E}[h \mid b, a] - \varepsilon \mathbb{E}[h \mid a])
    \\
    p_{\varepsilon}(c,a)
    & = p_0(c, a) \int p_{0}(b \mid c,a) (1 + \varepsilon \mathbb{E}[h \mid b, a] - \varepsilon \mathbb{E}[h \mid a]) db
    \\
    & = p_0(c, a) \int p_{0}(b \mid a) (1 + \varepsilon \mathbb{E}[h \mid b, a] - \varepsilon \mathbb{E}[h \mid a]) db
    & (C \Perp B \mid A)
    \\
    & = p_0(c, a)
\end{aligned}
\end{equation}
Therefore, we get $p_{\varepsilon}(c \mid b,a) = p_{0}(c \mid b,a) = p_{0}(c \mid a) = p_{\varepsilon}(c \mid a)$. Next,
\begin{equation}
\begin{aligned}
    p_{\varepsilon}(d,c,b)
    & = p_{0}(d,c,b) \int p_0(a \mid c,b) (1 + \varepsilon \mathbb{E}[h \mid b, a] - \varepsilon \mathbb{E}[h \mid a])  da
    & (A \Perp D \mid C,B)
    \\
    p_{\varepsilon}(c,b)
    & = \int p_{0}(c,b,a) (1 + \varepsilon \mathbb{E}[h \mid b, a] - \varepsilon \mathbb{E}[h \mid a]) da
    \\
    & = p_{0}(c,b) \int p_0(a \mid c,b) (1 + \varepsilon \mathbb{E}[h \mid b, a] - \varepsilon \mathbb{E}[h \mid a]) da
\end{aligned}
\end{equation}
This shows $p_{\varepsilon}(d \mid c,b,a) = p_{0}(d \mid c,b,a) = p_{0}(d \mid c,b) = p_{\varepsilon}(d \mid c,b)$. Hence $\mathcal{T}_{B|A} \subseteq \mathcal{T}$.

\textbf{Third case $\mathcal{T}_{C \mid A} \subseteq \mathcal{T}$:} Similar to $\mathcal{T}_{B|A} \subseteq \mathcal{T}$.

\textbf{Fourth case $\mathcal{T}_{D \mid CB} \subseteq \mathcal{T}$:} Pick $f = \mathbb{E}[h \mid d,c,b] - \mathbb{E}[h \mid c,b]$. Then
\begin{equation}
\begin{aligned}
    p_{\varepsilon}(c,b,a)
    & = p_{0}(c,b,a) \int p_{0}(d \mid c,b) (1 + \varepsilon \mathbb{E}[h \mid d,c,b] - \varepsilon \mathbb{E}[h \mid c,b]) dd
    & (D \Perp A \mid C,B)
    \\
    & = p_{0}(c,b,a)
\end{aligned}
\end{equation}
This show $p_{\varepsilon}(c \mid b,a) = p_{0}(c \mid b,a) = p_{0}(c \mid a) = p_{\varepsilon}(c \mid a)$. Moreover, $p_{\varepsilon}(d \mid c,b,a) = p_{0}(d \mid c,b,a) \times \int f p_0(d \mid c,b) dd$. Next
\begin{equation}
\begin{aligned}
    p_{\varepsilon}(d,c,b)
    & = p_{0}(d,c,b) \int p_0(a \mid d,c,b) (1 + \varepsilon \mathbb{E}[h \mid d,c,b] - \varepsilon \mathbb{E}[h \mid c,b])  da
    \\
    & = p_{0}(d,c,b) (1 + \varepsilon \mathbb{E}[h \mid d,c,b] - \varepsilon \mathbb{E}[h \mid c,b])
    \\
    p_{\varepsilon}(c,b)
    & = p_{0}(c,b) \int p_0(d \mid c,b) (1 + \varepsilon \mathbb{E}[h \mid d,c,b] - \varepsilon \mathbb{E}[h \mid c,b]) dd
    \\
    & = p_{0}(c,b)
\end{aligned}
\end{equation}
This shows $p_{\varepsilon}(d \mid c,b) = p_{0}(d \mid c,b) \times \int f p_0(d \mid c,b) dd$. Therefore $p_{\varepsilon}(d \mid c,b,a) = p_{\varepsilon}(d \mid c,b)$. Hence $\mathcal{T}_{D|CB} \subseteq \mathcal{T}$.

The proof is established because $\mathcal{T}_A, \mathcal{T}_{B \mid A}, \mathcal{T}_{C \mid A}, \mathcal{T}_{D \mid CB} \subseteq \mathcal{T}$ implies $\mathcal{T}_A \oplus \mathcal{T}_{B \mid A} \oplus \mathcal{T}_{C \mid A} \oplus \mathcal{T}_{D \mid CB} \subseteq \mathcal{T}$, and the equality follows.
\end{proof}

We reproduce a special case of the well-known result about DAG model's tangent space and its orthogonal complement \citep{Tsiatis.2006.SemiparametricTheory, Bhattacharya.Nabi.ea.2022.SemiparametricInference, Rotnitzky.Smucler.2020.EfficientAdjustment} in the following lemma.

\begin{lma}{\ref{lm:single_dag_model_orthocomp}}
    Suppose $\mathscr{P}$ is a single independence DAG model over variables $\mathbf{V}$, defined by a single constraint $\mathbf{X} \Perp \mathbf{Y} \mid \mathbf{Z}$. The model's tangent space $\mathcal{T}$ and its orthogonal complement $\mathcal{T}^\perp$ are, respectively
    \begin{equation*}
    \begin{aligned}
        \mathcal{T} = & \big\{ h - \mathbb{E}[h \mid \mathbf{x}, \mathbf{y}, \mathbf{z}]
        + \mathbb{E}[h \mid \mathbf{x}, \mathbf{z}] - \mathbb{E}[h \mid \mathbf{z}]
        \\
        & + \mathbb{E}[h \mid \mathbf{y}, \mathbf{z}]: \forall h \in \mathcal{H} \big\},
        & & \eqref{eq:single_dag_model_orthocomp}
        \\
        \mathcal{T}^\perp = & \big\{ \mathbb{E}[h \mid \mathbf{x}, \mathbf{y}, \mathbf{z}] - \mathbb{E}[h \mid \mathbf{y}, \mathbf{z}]
        \\
        & - \mathbb{E}[h \mid \mathbf{x}, \mathbf{z}] + \mathbb{E}[h \mid \mathbf{z}]
        : \forall h \in \mathcal{H} \big\}.
    \end{aligned}
    \end{equation*}
\end{lma}

\begin{proof}
Let $\mathbf{W} = \mathbf{V} \setminus (\mathbf{X} \dot{\cup} \mathbf{Y} \dot{\cup} \mathbf{Z})$. For any parametric submodel $\{p_{\varepsilon}(\mathbf{v})\}$ of $\mathscr{P}$, its elements factorize as
\begin{equation}
    \label{eq:lemma_factorization}
    p_{\varepsilon}(\mathbf{v}) = p_{\varepsilon}(\mathbf{w} \mid \mathbf{x}, \mathbf{y}, \mathbf{z}) p_{\varepsilon}(\mathbf{x} \mid \mathbf{z}) p_{\varepsilon}(\mathbf{y} \mid \mathbf{z}) p_{\varepsilon}(\mathbf{z})
\end{equation}
The total score function is $p_{\varepsilon}(\mathbf{v}) = \frac{\partial}{\partial \varepsilon} p_{\varepsilon}(\mathbf{v})|_{\varepsilon=0}$. By Lemma~\ref{lm:tangent_fact_1},
\begin{align}
    s(\mathbf{w} \mid \mathbf{x}, \mathbf{y}, \mathbf{z})
    & := \frac{\partial}{\partial \varepsilon} p_{\varepsilon}(\mathbf{w} \mid \mathbf{x}, \mathbf{y}, \mathbf{z}) |_{\varepsilon=0}
    = s(\mathbf{V}) - \EE[s(\mathbf{V}) \mid \mathbf{x}, \mathbf{y}, \mathbf{z}]
    \\
    s(\mathbf{x} \mid \mathbf{z})
    & := \frac{\partial}{\partial \varepsilon} p_{\varepsilon}(\mathbf{x} \mid \mathbf{z}) |_{\varepsilon=0}
    = \EE[s(\mathbf{V}) \mid \mathbf{x}, \mathbf{z}] - \EE[s(\mathbf{V}) \mid \mathbf{z}]
    \\
    s(\mathbf{y} \mid \mathbf{z})
    & := \frac{\partial}{\partial \varepsilon} p_{\varepsilon}(\mathbf{y} \mid \mathbf{z}) |_{\varepsilon=0}
    = \EE[s(\mathbf{V}) \mid \mathbf{y}, \mathbf{z}] - \EE[s(\mathbf{V}) \mid \mathbf{z}]
    \\
    s(\mathbf{z})
    & := \frac{\partial}{\partial \varepsilon} p_{\varepsilon}(\mathbf{z}) |_{\varepsilon=0}
    = \EE[s(\mathbf{V}) \mid \mathbf{z}]
\end{align}
Define the following subspaces
\begin{align}
    \mathcal{T}_{\mathbf{W} \mid \mathbf{XYZ}}
    & := \left\{ h - \EE[h \mid \mathbf{x}, \mathbf{y}, \mathbf{z}] : h \in \mathcal{H} \right\}
    \\
    \mathcal{T}_{\mathbf{X} \mid \mathbf{Z}}
    & := \left\{ \EE[h \mid \mathbf{x}, \mathbf{z}] - \EE[h \mid \mathbf{z}] : h \in \mathcal{H} \right\}
    \\
    \mathcal{T}_{\mathbf{Y} \mid \mathbf{Z}}
    & := \left\{ \EE[h \mid \mathbf{y}, \mathbf{z}] - \EE[h \mid \mathbf{z}] : h \in \mathcal{H} \right\}
    \\
    \mathcal{T}_{\mathbf{Z}}
    & := \left\{ \EE[h \mid \mathbf{z}] : h \in \mathcal{H} \right\}
\end{align}
It is evident that $s(\mathbf{w} \mid \mathbf{x}, \mathbf{y}, \mathbf{z}) \in \mathcal{T}_{\mathbf{W} \mid \mathbf{XYZ}}$, $s(\mathbf{x} \mid \mathbf{z}) \in \mathcal{T}_{\mathbf{X} \mid \mathbf{Z}}$, $s(\mathbf{y} \mid \mathbf{z}) \in \mathcal{T}_{\mathbf{Y} \mid \mathbf{Z}}$ and $s(\mathbf{z}) \in \mathcal{T}_{\mathbf{Z}}$. By Lemma~\ref{lm:tangent_fact_2}, these subspaces are orthogonal. Moreover, due to Equation~\ref{eq:lemma_factorization}
\begin{equation}
    s(\mathbf{v}) = s(\mathbf{w} \mid \mathbf{x}, \mathbf{y}, \mathbf{z}) + s(\mathbf{x} \mid \mathbf{z}) + s(\mathbf{y} \mid \mathbf{z}) + s(\mathbf{z})
\end{equation}
Then $s(\mathbf{v}) \in \mathcal{T}_{\mathbf{W} \mid \mathbf{XYZ}} \oplus \mathcal{T}_{\mathbf{X} \mid \mathbf{Z}} \oplus \mathcal{T}_{\mathbf{Y} \mid \mathbf{Z}} \oplus \mathcal{T}_{\mathbf{Z}}$. This means the tangent space of this model is $\mathcal{T} \subseteq \mathcal{T}_{\mathbf{W} \mid \mathbf{XYZ}} \oplus \mathcal{T}_{\mathbf{X} \mid \mathbf{Z}} \oplus \mathcal{T}_{\mathbf{Y} \mid \mathbf{Z}} \oplus \mathcal{T}_{\mathbf{Z}}$.

Next, we will show that each of the subspace is in $\mathcal{T}$. To do this, we must show that any function in these subspaces can be used to construct parametric submodel of $\mathscr{P}$.

Pick any $f = h - \EE[h \mid \mathbf{x}, \mathbf{y}, \mathbf{z}] \in \mathcal{T}_{\mathbf{W} \mid \mathbf{XYZ}}$ and construct a parametric model by $p_{\varepsilon} (\mathbf{v}) = p_{0} (\mathbf{v})(1 + \varepsilon f(\mathbf{v}))$. Since $\mathbb{E}[f] = 0$, $p_{\varepsilon}$ sums to 1. Moreover, the parameter $\varepsilon \in (- \delta, \delta)$ and $\delta$ is small enough so that $p_{\varepsilon} (\mathbf{v}) \geq 0$ at all $\mathbf{v}$. Then $p_{\varepsilon}$ is a valid probability distribution. This parametric model contains $p_0$ at $\varepsilon=0$, and the score function at $p_0$ is $\frac{\partial}{\partial \varepsilon} \log p_{\varepsilon} (\mathbf{v})|_{\varepsilon=0} = f$. What left to show is that the constraint holds for all $p_{\varepsilon}$. Indeed, for any $p_{\varepsilon}$,
\begin{equation}
     p_{\varepsilon}(\mathbf{x}, \mathbf{y}, \mathbf{z}) = p_{0} (\mathbf{x}, \mathbf{y}, \mathbf{z}) \left(1 + \varepsilon \int f(\mathbf{x}, \mathbf{y}, \mathbf{z}, \mathbf{w}) p_0(\mathbf{w} \mid \mathbf{x}, \mathbf{y}, \mathbf{z}) d\mathbf{w} \right) = p_{0} (\mathbf{x}, \mathbf{y}, \mathbf{z}),
\end{equation}
since $\EE[f \mid \mathbf{x}, \mathbf{y}, \mathbf{z}] = 0$. Therefore the constraint $\mathbf{X} \Perp \mathbf{Y} \mid \mathbf{Z}$ holds for $p_{\varepsilon}$, $p_{\varepsilon}(\mathbf{x} \mid \mathbf{y}, \mathbf{z}) = p_{\varepsilon}(\mathbf{x} \mid \mathbf{z})$, because it holds for $p_0$. This means $\{p_{\varepsilon} (\mathbf{v}): \varepsilon \in (- \delta, \delta)\}$ is also a parametric submodel of $\mathscr{P}$, hence $f \in \mathcal{T}$ and therefore $\mathcal{T}_{\mathbf{W} \mid \mathbf{XYZ}} \subseteq \mathcal{T}$.

Pick any $f = \EE[h \mid \mathbf{x}, \mathbf{z}]  - \EE[h \mid \mathbf{z}] \in \mathcal{T}_{\mathbf{X} \mid \mathbf{Z}}$ and construct a parametric model by $p_{\varepsilon} (\mathbf{v}) = p_{0} (\mathbf{v})(1 + \varepsilon f(\mathbf{v}))$. By the same reasoning, all we need to show is the constraint holds for all $p_{\varepsilon} (\mathbf{v})$.
\begin{equation}
\begin{aligned}
     p_{\varepsilon}(\mathbf{x}, \mathbf{y}, \mathbf{z}) & = p_{0} (\mathbf{x}, \mathbf{y}, \mathbf{z}) (1 + \varepsilon \EE [f \mid \mathbf{x}, \mathbf{y}, \mathbf{z}])
     = p_{0} (\mathbf{x}, \mathbf{y}, \mathbf{z}) (1 + \varepsilon \EE[h \mid \mathbf{x}, \mathbf{z}]  - \varepsilon \EE[h \mid \mathbf{z}])
     \\
     p_{\varepsilon}(\mathbf{y}, \mathbf{z}) & = p_{0} (\mathbf{y}, \mathbf{z}) (1 + \varepsilon \EE [f \mid \mathbf{y}, \mathbf{z}])
     = p_{0} (\mathbf{y}, \mathbf{z}) \bigg(1 + \varepsilon \int \EE[h \mid \mathbf{x}, \mathbf{z}] p_0(\mathbf{x} \mid \mathbf{y}, \mathbf{z}) d\mathbf{x}  - \varepsilon \EE[h \mid \mathbf{z}]\bigg)
     \\
     & = p_{0} (\mathbf{y}, \mathbf{z}) \bigg(1 + \varepsilon \int \EE[h \mid \mathbf{x}, \mathbf{z}] p_0(\mathbf{x} \mid \mathbf{z}) d\mathbf{x}  - \varepsilon \EE[h \mid \mathbf{z}]\bigg)
     \quad \quad (\mathbf{X} \Perp \mathbf{Y} \mid \mathbf{Z} \text{ at }p_0)
     \\
     & = p_{0} (\mathbf{y}, \mathbf{z})
     \\
     p_{\varepsilon}(\mathbf{x}, \mathbf{z}) & = p_{0} (\mathbf{x}, \mathbf{z}) (1 + \varepsilon \EE [f \mid \mathbf{x}, \mathbf{z}])
     = p_{0} (\mathbf{x}, \mathbf{z}) (1 + \varepsilon \EE[h \mid \mathbf{x}, \mathbf{z}]  - \varepsilon \EE[h \mid \mathbf{z}])
     \\
     p_{\varepsilon}(\mathbf{z}) & = p_{0} (\mathbf{z}) (1 + \varepsilon \EE [f \mid \mathbf{z}]) = p_{0} (\mathbf{z})
     \quad \quad \quad \quad (\text{since } \EE [f \mid \mathbf{z}] = 0)
\end{aligned}
\end{equation}
Therefore,
\begin{equation}
\begin{aligned}
    p_{\varepsilon}(\mathbf{x} \mid \mathbf{y}, \mathbf{z})
    & := \frac{p_{\varepsilon}(\mathbf{x}, \mathbf{y}, \mathbf{z})}{p_{\varepsilon}(\mathbf{y}, \mathbf{z})}
    = p_{0} (\mathbf{x} \mid \mathbf{y}, \mathbf{z}) (1 + \varepsilon \EE[h \mid \mathbf{x}, \mathbf{z}]  - \varepsilon \EE[h \mid \mathbf{z}])
    \\
    & = p_{0} (\mathbf{x} \mid \mathbf{z}) (1 + \varepsilon \EE[h \mid \mathbf{x}, \mathbf{z}]  - \varepsilon \EE[h \mid \mathbf{z}])
    = \frac{p_{\varepsilon}(\mathbf{x}, \mathbf{z})}{p_{\varepsilon}(\mathbf{z})} = p_{\varepsilon}(\mathbf{x} \mid \mathbf{z})
\end{aligned}
\end{equation}
Hence, the constraint holds for all $p_{\varepsilon}$, so $\{p_{\varepsilon}\}$ is a parametric submodel of $\mathscr{P}$, and therefore $f \in \mathcal{T}_{\mathbf{X} \mid \mathbf{Z}} \subseteq \mathcal{T}$.

Pick any $f = \EE[h \mid \mathbf{y}, \mathbf{z}]  - \EE[h \mid \mathbf{z}] \in \mathcal{T}_{\mathbf{X} \mid \mathbf{Z}}$ and construct a parametric model by $p_{\varepsilon} (\mathbf{v}) = p_{0} (\mathbf{v})(1 + \varepsilon f(\mathbf{v}))$. By the same reasoning, all we need to show is the constraint holds for all $p_{\varepsilon} (\mathbf{v})$.
\begin{equation}
\begin{aligned}
     p_{\varepsilon}(\mathbf{x}, \mathbf{y}, \mathbf{z}) & = p_{0} (\mathbf{x}, \mathbf{y}, \mathbf{z}) (1 + \varepsilon \EE [f \mid \mathbf{x}, \mathbf{y}, \mathbf{z}])
     = p_{0} (\mathbf{x}, \mathbf{y}, \mathbf{z}) (1 + \varepsilon \EE[h \mid \mathbf{y}, \mathbf{z}]  - \varepsilon \EE[h \mid \mathbf{z}])
     \\
     p_{\varepsilon}(\mathbf{y}, \mathbf{z}) & = p_{0} (\mathbf{y}, \mathbf{z}) (1 + \varepsilon \EE [f \mid \mathbf{y}, \mathbf{z}])
     = p_{0} (\mathbf{y}, \mathbf{z}) (1 + \varepsilon \EE[h \mid \mathbf{y}, \mathbf{z}]  - \varepsilon \EE[h \mid \mathbf{z}])
     \\
     p_{\varepsilon}(\mathbf{x}, \mathbf{z}) & = p_{0} (\mathbf{x}, \mathbf{z}) (1 + \varepsilon \EE [f \mid \mathbf{x}, \mathbf{z}])
     = p_{0} (\mathbf{x}, \mathbf{z}) \bigg(1 + \varepsilon \int \EE[h \mid \mathbf{y}, \mathbf{z}] p_0(\mathbf{y} \mid \mathbf{x}, \mathbf{z}) d\mathbf{y}  - \varepsilon \EE[h \mid \mathbf{z}]\bigg)
     \\
     & = p_{0} (\mathbf{x}, \mathbf{z}) \bigg(1 + \varepsilon \int \EE[h \mid \mathbf{y}, \mathbf{z}] p_0(\mathbf{y} \mid \mathbf{z}) d\mathbf{x}  - \varepsilon \EE[h \mid \mathbf{z}]\bigg)
     \quad \quad (\mathbf{Y} \Perp \mathbf{X} \mid \mathbf{Z} \text{ at }p_0)
     \\
     & = p_{0} (\mathbf{x}, \mathbf{z})
     \\
     p_{\varepsilon}(\mathbf{z}) & = p_{0} (\mathbf{z}) (1 + \varepsilon \EE [f \mid \mathbf{z}]) = p_{0} (\mathbf{z})
     \quad \quad \quad \quad (\text{since } \EE [f \mid \mathbf{z}] = 0)
\end{aligned}
\end{equation}
Therefore,
\begin{equation}
\begin{aligned}
    p_{\varepsilon}(\mathbf{x} \mid \mathbf{y}, \mathbf{z})
    & := \frac{p_{\varepsilon}(\mathbf{x}, \mathbf{y}, \mathbf{z})}{p_{\varepsilon}(\mathbf{y}, \mathbf{z})}
    = p_{0} (\mathbf{x} \mid \mathbf{y}, \mathbf{z}) = p_{0}(\mathbf{x} \mid \mathbf{z}) = p_{\varepsilon}(\mathbf{x} \mid \mathbf{z})
\end{aligned}
\end{equation}
Hence, the constraint holds for all $p_{\varepsilon}$, so $\{p_{\varepsilon}\}$ is a parametric submodel of $\mathscr{P}$, and therefore $f \in \mathcal{T}_{\mathbf{Y} \mid \mathbf{Z}} \subseteq \mathcal{T}$.

Pick any $f = \EE[h \mid \mathbf{z}] \in \mathcal{T}_{\mathbf{Z}}$ and construct a parametric model by $p_{\varepsilon} (\mathbf{v}) = p_{0} (\mathbf{v})(1 + \varepsilon f(\mathbf{v}))$. Then
\begin{equation}
\begin{aligned}
     p_{\varepsilon}(\mathbf{x}, \mathbf{y}, \mathbf{z}) & = p_{0} (\mathbf{x}, \mathbf{y}, \mathbf{z}) (1 + \varepsilon \EE [f \mid \mathbf{x}, \mathbf{y}, \mathbf{z}])
     = p_{0} (\mathbf{x}, \mathbf{y}, \mathbf{z}) (1 + \varepsilon \EE[h \mid \mathbf{z}])
     \\
     p_{\varepsilon}(\mathbf{y}, \mathbf{z}) & = p_{0} (\mathbf{y}, \mathbf{z}) (1 + \varepsilon \EE [f \mid \mathbf{y}, \mathbf{z}])
     = p_{0} (\mathbf{y}, \mathbf{z}) (1 + \varepsilon \EE[h \mid \mathbf{z}])
     \\
     p_{\varepsilon}(\mathbf{x}, \mathbf{z}) & = p_{0} (\mathbf{x}, \mathbf{z}) (1 + \varepsilon \EE [f \mid \mathbf{x}, \mathbf{z}])
     = p_{0} (\mathbf{x}, \mathbf{z}) (1 + \varepsilon \EE[h \mid \mathbf{z}])
     \\
     p_{\varepsilon}(\mathbf{z}) & = p_{0} (\mathbf{z}) (1 + \varepsilon \EE [f \mid \mathbf{z}]) = = p_{0} (\mathbf{z}) (1 + \varepsilon \EE[h \mid \mathbf{z}])
\end{aligned}
\end{equation}
Therefore,
\begin{equation}
\begin{aligned}
    p_{\varepsilon}(\mathbf{x} \mid \mathbf{y}, \mathbf{z})
    & := \frac{p_{\varepsilon}(\mathbf{x}, \mathbf{y}, \mathbf{z})}{p_{\varepsilon}(\mathbf{y}, \mathbf{z})}
    = p_{0} (\mathbf{x} \mid \mathbf{y}, \mathbf{z}) = p_{0}(\mathbf{x} \mid \mathbf{z}) = p_{\varepsilon}(\mathbf{x} \mid \mathbf{z})
\end{aligned}
\end{equation}
Hence, the constraint holds for all $p_{\varepsilon}$, so $\{p_{\varepsilon}\}$ is a parametric submodel of $\mathscr{P}$, and therefore $f \in \mathcal{T}_{\mathbf{Z}} \subseteq \mathcal{T}$.

These results imply $\mathcal{T}_{\mathbf{W} \mid \mathbf{XYZ}} \oplus \mathcal{T}_{\mathbf{X} \mid \mathbf{Z}} \oplus \mathcal{T}_{\mathbf{Y} \mid \mathbf{Z}} \oplus \mathcal{T}_{\mathbf{Z}} \subseteq \mathcal{T}$, and with previous result, we get $\mathcal{T} = \mathcal{T}_{\mathbf{W} \mid \mathbf{XYZ}} \oplus \mathcal{T}_{\mathbf{X} \mid \mathbf{Z}} \oplus \mathcal{T}_{\mathbf{Y} \mid \mathbf{Z}} \oplus \mathcal{T}_{\mathbf{Z}}$.

Define the following operators
\begin{align}
    \Pi(h \mid \mathcal{T}_{\mathbf{W} \mid \mathbf{XYZ}}) & = h - \EE[h \mid \mathbf{x}, \mathbf{y}, \mathbf{z}]
    \\
    \Pi(h \mid \mathcal{T}_{\mathbf{X} \mid \mathbf{Z}}) & = \EE[h \mid \mathbf{x}, \mathbf{z}] - \EE[h \mid \mathbf{z}]
    \\
    \Pi(h \mid \mathcal{T}_{\mathbf{Y} \mid \mathbf{Z}}) & = \EE[h \mid \mathbf{y}, \mathbf{z}] - \EE[h \mid \mathbf{z}]
    \\
    \Pi(h \mid \mathcal{T}_{\mathbf{Z}}) & = \EE[h \mid \mathbf{z}]
\end{align}
By Lemma~\ref{lm:tangent_fact_3}, $\Pi(\cdot \mid \mathcal{T}_{\mathbf{W} \mid \mathbf{XYZ}}), \Pi(\cdot \mid \mathcal{T}_{\mathbf{X} \mid \mathbf{Z}}), \Pi(\cdot \mid \mathcal{T}_{\mathbf{Y} \mid \mathbf{Z}}), \Pi(\cdot \mid \mathcal{T}_{\mathbf{Z}})$ are the orthogonal projection operators onto $\mathcal{T}_{\mathbf{W} \mid \mathbf{XYZ}}, \mathcal{T}_{\mathbf{X} \mid \mathbf{Z}}, \mathcal{T}_{\mathbf{Y} \mid \mathbf{Z}}, \mathcal{T}_{\mathbf{Z}}$, respectively.

Define the following operator
\begin{equation}
\begin{aligned}
    \Pi(h)
    & = \Pi(h \mid \mathcal{T}_{\mathbf{W} \mid \mathbf{XYZ}}) + \Pi(h \mid \mathcal{T}_{\mathbf{X} \mid \mathbf{Z}}) + \Pi(h \mid \mathcal{T}_{\mathbf{Y} \mid \mathbf{Z}}) + \Pi(h \mid \mathcal{T}_{\mathbf{Z}})
    \\
    & = h - \EE[h \mid \mathbf{x}, \mathbf{y}, \mathbf{z}] + \EE[h \mid \mathbf{x}, \mathbf{z}] - \EE[h \mid \mathbf{z}] + \EE[h \mid \mathbf{y}, \mathbf{z}]
\end{aligned}
\end{equation}
We will show that $\Pi$ is the projection operator onto the tangent space $\mathcal{T}$.

First, we will show $\mathcal{T} = \{\Pi(h) : \forall h \in \mathcal{H}\}$. For the first direction, $\Pi(h) \in \mathcal{T}$ for all $h \in \mathcal{H}$, obviously. For the other direction, since $\mathcal{T} = \mathcal{T}_{\mathbf{W} \mid \mathbf{XYZ}} \oplus \mathcal{T}_{\mathbf{X} \mid \mathbf{Z}} \oplus \mathcal{T}_{\mathbf{Y} \mid \mathbf{Z}} \oplus \mathcal{T}_{\mathbf{Z}}$, any $f \in \mathcal{T}$ can be written as, for some $h_1, h_2, h_3, h_4 \in \mathcal{H}$,
\begin{equation}
    f = h_1 - \EE[h_1 \mid \mathbf{x}, \mathbf{y}, \mathbf{z}] + \EE[h_2 \mid \mathbf{x}, \mathbf{z}] - \EE[h_2 \mid \mathbf{z}] + \EE[h_3 \mid \mathbf{y}, \mathbf{z}] - \EE[h_3 \mid \mathbf{z}] + \EE[h_4 \mid \mathbf{z}].
\end{equation}
Because the subspaces $\mathcal{T}_{\mathbf{W} \mid \mathbf{XYZ}}, \mathcal{T}_{\mathbf{X} \mid \mathbf{Z}}, \mathcal{T}_{\mathbf{Y} \mid \mathbf{Z}}, \mathcal{T}_{\mathbf{Z}}$ are orthogonal,
\begin{equation}
\begin{aligned}
    \Pi(f \mid \mathcal{T}_{\mathbf{W} \mid \mathbf{XYZ}}) & = h_1 - \EE[h_1 \mid \mathbf{x}, \mathbf{y}, \mathbf{z}]
    \\
    \Pi(f \mid \mathcal{T}_{\mathbf{X} \mid \mathbf{Z}}) & = \EE[h_2 \mid \mathbf{x}, \mathbf{z}] - \EE[h_2 \mid \mathbf{z}]
    \\
    \Pi(f \mid \mathcal{T}_{\mathbf{Y} \mid \mathbf{Z}}) & = \EE[h_3 \mid \mathbf{y}, \mathbf{z}] - \EE[h_3 \mid \mathbf{z}]
    \\
    \Pi(f \mid \mathcal{T}_{\mathbf{Z}}) & = \EE[h_4 \mid \mathbf{z}].
\end{aligned}
\end{equation}
Therefore, $f = \Pi(f)$, hence $\mathcal{T} \subseteq \{\Pi(h) : h \in \mathcal{H}\}$. This implies $\mathcal{T} = \{\Pi(h) : h \in \mathcal{H}\}$.

Second, we will show that $f - \Pi(f)$ and $\Pi(g)$ are orthogonal, for all $f, g \in \mathcal{H}$.
\begin{equation}
\begin{aligned}
    & \EE \bigg[f - \Pi(f), \Pi(g) \bigg]
    = \EE \bigg[f \Pi(g) \bigg] - \EE \bigg[\Pi(f) \Pi(g) \bigg]
    \\
    & = \EE \bigg[ \bigg( \EE[f \mid \mathbf{X}, \mathbf{Y}, \mathbf{Z}] - \EE[f \mid \mathbf{X}, \mathbf{Z}] + \EE[f \mid \mathbf{Z}] - \EE[f \mid \mathbf{Y}, \mathbf{Z}] \bigg)
    \times \bigg( g - \EE[g \mid \mathbf{X}, \mathbf{Y}, \mathbf{Z}] + \EE[g \mid \mathbf{X}, \mathbf{Z}] - \EE[g \mid \mathbf{Z}] + \EE[g \mid \mathbf{Y}, \mathbf{Z}] \bigg) \bigg]
    \\
    & = \EE \bigg[ \bigg( \EE[f \mid \mathbf{X}, \mathbf{Y}, \mathbf{Z}] - \EE[f \mid \mathbf{X}, \mathbf{Z}] + \EE[f \mid \mathbf{Z}] - \EE[f \mid \mathbf{Y}, \mathbf{Z}] \bigg)
    \times \bigg( \EE[g \mid \mathbf{X}, \mathbf{Z}] - \EE[g \mid \mathbf{Z}] + \EE[g \mid \mathbf{Y}, \mathbf{Z}] \bigg) \bigg]
    \\
    & = \EE \big[ \EE[f \mid \mathbf{X}, \mathbf{Y}, \mathbf{Z}] \EE[g \mid \mathbf{X}, \mathbf{Z}] \big] - \EE \big[ \EE[f \mid \mathbf{X}, \mathbf{Z}] \EE[g \mid \mathbf{X}, \mathbf{Z}] \big] + \EE \big[ \EE[f \mid \mathbf{Z}] \EE[g \mid \mathbf{X}, \mathbf{Z}] \big] - \EE \big[ \EE[f \mid \mathbf{Y}, \mathbf{Z}] \EE[g \mid \mathbf{X}, \mathbf{Z}] \big]
    \\
    & - \EE \big[ \EE[f \mid \mathbf{X}, \mathbf{Y}, \mathbf{Z}] \EE[g \mid \mathbf{Z}] \big] + \EE \big[ \EE[f \mid \mathbf{X}, \mathbf{Z}] \EE[g \mid \mathbf{Z}] \big] - \EE \big[ \EE[f \mid \mathbf{Z}] \EE[g \mid \mathbf{Z}] \big] + \EE \big[ \EE[f \mid \mathbf{Y}, \mathbf{Z}] \EE[g \mid \mathbf{Z}] \big]
    \\
    & + \EE \big[ \EE[f \mid \mathbf{X}, \mathbf{Y}, \mathbf{Z}] \EE[g \mid \mathbf{Y}, \mathbf{Z}] \big] - \EE \big[ \EE[f \mid \mathbf{X}, \mathbf{Z}] \EE[g \mid \mathbf{Y}, \mathbf{Z}] \big] + \EE \big[ \EE[f \mid \mathbf{Z}] \EE[g \mid \mathbf{Y}, \mathbf{Z}] \big] - \EE \big[ \EE[f \mid \mathbf{Y}, \mathbf{Z}] \EE[g \mid \mathbf{Y}, \mathbf{Z}] \big]
    \\
    & = \EE \big[ \EE[f \mid \mathbf{X}, \mathbf{Z}] \EE[g \mid \mathbf{X}, \mathbf{Z}] \big] - \EE \big[ \EE[f \mid \mathbf{X}, \mathbf{Z}] \EE[g \mid \mathbf{X}, \mathbf{Z}] \big] + \EE \big[ \EE[f \mid \mathbf{Z}] \EE[g \mid \mathbf{Z}] \big] - \EE \big[ \EE[f \mid \mathbf{Y}, \mathbf{Z}] \EE[g \mid \mathbf{X}, \mathbf{Z}] \big]
    \\
    & - \EE \big[ \EE[f \mid \mathbf{Z}] \EE[g \mid \mathbf{Z}] \big] + \EE \big[ \EE[f \mid \mathbf{Z}] \EE[g \mid \mathbf{Z}] \big] - \EE \big[ \EE[f \mid \mathbf{Z}] \EE[g \mid \mathbf{Z}] \big] + \EE \big[ \EE[f \mid \mathbf{Z}] \EE[g \mid \mathbf{Z}] \big]
    \\
    & + \EE \big[ \EE[f \mid \mathbf{Y}, \mathbf{Z}] \EE[g \mid \mathbf{Y}, \mathbf{Z}] \big] - \EE \big[ \EE[f \mid \mathbf{X}, \mathbf{Z}] \EE[g \mid \mathbf{Y}, \mathbf{Z}] \big] + \EE \big[ \EE[f \mid \mathbf{Z}] \EE[g \mid \mathbf{Z}] \big] - \EE \big[ \EE[f \mid \mathbf{Y}, \mathbf{Z}] \EE[g \mid \mathbf{Y}, \mathbf{Z}] \big]
    \\
    & = - \EE \big[ \EE[f \mid \mathbf{Y}, \mathbf{Z}] \EE[g \mid \mathbf{X}, \mathbf{Z}] \big] - \EE \big[ \EE[f \mid \mathbf{X}, \mathbf{Z}] \EE[g \mid \mathbf{Y}, \mathbf{Z}] \big] + 2 \EE \big[ \EE[f \mid \mathbf{Z}] \EE[g \mid \mathbf{Z}] \big]
\end{aligned}
\end{equation}
Due to the constraint
\begin{equation}
\begin{aligned}
    \EE \big[ \EE[f \mid \mathbf{Y}, \mathbf{Z}] \EE[g \mid \mathbf{X}, \mathbf{Z}] \big]
    & = \int \EE[f \mid \mathbf{y}, \mathbf{z}] \EE[g \mid \mathbf{x}, \mathbf{z}] p_0(\mathbf{x}, \mathbf{y}, \mathbf{z}) d\mathbf{x} d\mathbf{y} d\mathbf{z}
    \\
    & = \int \EE[f \mid \mathbf{y}, \mathbf{z}] \EE[g \mid \mathbf{x}, \mathbf{z}] p_0(\mathbf{x} \mid \mathbf{z}) p_0(\mathbf{y} \mid \mathbf{z}) p_0(\mathbf{z}) d\mathbf{x} d\mathbf{y} d\mathbf{z}
    \\
    & = \int \EE[f \mid \mathbf{z}] \EE[g \mid \mathbf{z}] p_0(\mathbf{z}) d\mathbf{z}
    \\
    & = \EE \big[ \EE[f \mid \mathbf{Z}] \EE[g \mid \mathbf{Z}] \big]
\end{aligned}
\end{equation}
Similarly, $\EE \big[ \EE[f \mid \mathbf{X}, \mathbf{Z}] \EE[g \mid \mathbf{Y}, \mathbf{Z}] \big] = \EE \big[ \EE[f \mid \mathbf{Z}] \EE[g \mid \mathbf{Z}] \big]$. Therefore, the inner product vanishes,\\ $\EE \bigg[f - \Pi(f), \Pi(g) \bigg] = 0$.

In conclusion, the operator defined by
\begin{equation}
    \Pi(h \mid \mathcal{T}) = h - \EE[h \mid \mathbf{x}, \mathbf{y}, \mathbf{z}] + \EE[h \mid \mathbf{x}, \mathbf{z}] - \EE[h \mid \mathbf{z}] + \EE[h \mid \mathbf{y}, \mathbf{z}]
\end{equation}
is the projection operator onto the tangent space of this model, hence $\mathcal{T} = \{\Pi(h \mid \mathcal{T}) : \forall h \in \mathcal{H}\}$. The operator onto its orthogonal complement is
\begin{equation}
\begin{aligned}
    \Pi(h \mid \mathcal{T}^\perp) & = h - \Pi(h \mid \mathcal{T})
    \\
    & = \EE[h \mid \mathbf{x}, \mathbf{y}, \mathbf{z}] - \EE[h \mid \mathbf{x}, \mathbf{z}] + \EE[h \mid \mathbf{z}] - \EE[h \mid \mathbf{y}, \mathbf{z}]
\end{aligned}
\end{equation}
and $\mathcal{T}^\perp = \{\Pi(h \mid \mathcal{T}^\perp) : \forall h \in \mathcal{H}\}$.

\end{proof}

\section{Proofs in Section \ref{sec:ci_tangent}}

\begin{thma}{\ref{thm:general_ci_model_orthocomp}}
    Suppose $\mathscr{P}$ is a general semi-parametric Markov model defined by $K \geq 1$ conditional independence constraints $\mathbf{X}_i \Perp \mathbf{Y}_i \mid \mathbf{Z}_i$, where $i = 1, \ldots, K$. For each $i$, let $\mathscr{P}_i$ be the single independence DAG model defined by the $i$-th constraint. The orthogonal projection operator $\Pi(\cdot \mid \mathcal{T}_i^\perp): \mathcal{H} \mapsto \mathcal{H}$ sends each function $h \in \mathcal{H}$ to the orthogonal complement $\mathcal{T}_i^\perp$ of the tangent space $\mathcal{T}_i$ of $\mathscr{P}_i$,
    \begin{equation}
    \begin{aligned}
        \Pi(h \mid \mathcal{T}_i^\perp) (\mathbf{x}_{i}, \mathbf{y}_{i}, \mathbf{z}_{i})
        = \mathbb{E}[h | \mathbf{x}_{i}, \mathbf{y}_{i}, \mathbf{z}_{i}] - \mathbb{E}[h | \mathbf{x}_{i}, \mathbf{z}_{i}]
        - \mathbb{E}[h | \mathbf{y}_{i}, \mathbf{z}_{i}] + \mathbb{E}[h | \mathbf{z}_{i}].
    \end{aligned}
    \end{equation}
    The orthogonal complement of the tangent space of $\mathscr{P}$ is the direct sum of these orthogonal complements,
    \begin{equation}
    \label{eq:dag_orthocomp_general}
    \begin{aligned}
        \mathcal{T}^{\perp} = \left\{ \sum_{i = 1}^K \Pi(h_i \mid \mathcal{T}_i^\perp) \: : \: \forall h_1, \ldots, h_K \in \mathcal{H} \right\}.
    \end{aligned}
    \end{equation}
\end{thma}

\begin{proof}
First, we will show $\mathcal{T} = \mathcal{T}_1 \cap \ldots \cap \mathcal{T}_K$.

Let $\{p_{\varepsilon}(\mathbf{v}) : \varepsilon \in (- \delta, \delta)\}$ be any parametric submodel of $\mathscr{P}$, with total score function $s(\mathbf{v})$ at $p_0$. For each $i$, $p_{\varepsilon}(\mathbf{v})$ satisfies the constraint $\mathbf{X}_i \Perp \mathbf{Y}_i \mid \mathbf{Z}_i$. Therefore, $\{p_{\varepsilon}(\mathbf{v}) : \varepsilon \in (- \delta, \delta)\}$ is also a parametric submodel of $\mathscr{P}_i$. Following the proof of Lemma~\ref{lm:single_dag_model_orthocomp}, we have $s \in \mathcal{T}_i$, which is the tangent space of $\mathscr{P}_i$. This is true for all $i=1, \ldots, K$, so $s \in \overline{\mathcal{T}_1 \cap \ldots \cap \mathcal{T}_K}$ and therefore $\mathcal{T} \subseteq \overline{\mathcal{T}_1 \cap \ldots \cap \mathcal{T}_K}$.

For the other direction, pick any $f \in \mathcal{T}_1 \cap \ldots \cap \mathcal{T}_K$. Define a parametric model by $p_{\varepsilon} (\mathbf{v}) = p_{0} (\mathbf{v})(1 + \varepsilon f(\mathbf{v}))$. Since $\mathbb{E}[f] = 0$, $p_{\varepsilon}$ sums to 1. Moreover, the parameter $\varepsilon \in (- \delta, \delta)$ and $\delta$ is small enough so that $p_{\varepsilon} (\mathbf{v}) \geq 0$ at all $\mathbf{v}$. Then $p_{\varepsilon}$ is a valid probability distribution. This parametric model contains $p_0$ at $\varepsilon=0$, and the score function at $p_0$ is $\frac{\partial}{\partial \varepsilon} \log p_{\varepsilon} (\mathbf{v})|_{\varepsilon=0} = f$. What left to show is that all $K$ constraints hold for all $p_{\varepsilon}(\mathbf{v})$. Indeed, for each $i$, $f \in \mathcal{T}_i$. Following the proof of Lemma~\ref{lm:single_dag_model_orthocomp}, the $i$-th constraint $\mathbf{X}_i \Perp \mathbf{Y}_i \mid \mathbf{Z}_i$ must hold in $p_{\varepsilon}(\mathbf{v})$. Therefore, $\mathcal{T}_1 \cap \ldots \cap \mathcal{T}_K \subseteq \mathcal{T}$. By definition, $\mathcal{T}$ is closed, so $\overline{\mathcal{T}_1 \cap \ldots \cap \mathcal{T}_K} \subseteq \mathcal{T}$.

We have shown that $\mathcal{T} = \overline{\mathcal{T}_1 \cap \ldots \cap \mathcal{T}_K}$. Since $\mathcal{T}_1, \ldots, \mathcal{T}_K$ are closed subspaces, $\mathcal{T} = \mathcal{T}_1 \cap \ldots \cap \mathcal{T}_K$.

Next, we will show $\mathcal{T}^\perp = \mathcal{T}_1^\perp \oplus \ldots \oplus \mathcal{T}_K^\perp$.

We have $\mathcal{T} = \mathcal{T}_1 \cap \ldots \cap \mathcal{T}_K$ is orthogonal to all $\{\mathcal{T}_i^\perp\}_{i=1}^K$, so $\mathcal{T}$ is orthogonal to the span $\overline{\mathcal{T}_1^\perp \oplus \ldots \oplus \mathcal{T}_K^\perp}$. Since $\mathcal{T}_1, \ldots, \mathcal{T}_K$ are closed, this span is exactly the direct sum $\mathcal{T}_1^\perp \oplus \ldots \oplus \mathcal{T}_K^\perp$. Therefore, $\mathcal{T}_1^\perp \oplus \ldots \oplus \mathcal{T}_K^\perp \subseteq \mathcal{T}^\perp$.

For the other direction, every $f \in \big( \mathcal{T}_1^\perp \oplus \ldots \oplus \mathcal{T}_K^\perp \big)^\perp$ must be orthogonal to all $\{\mathcal{T}_i^\perp\}_{i=1}^K$. For each $i$, the fact that $f$ is orthogonal to $\mathcal{T}_i^\perp$ implies $f \in \mathcal{T}_i$. Therefore, $f \in \mathcal{T}_1 \cap \ldots \cap \mathcal{T}_K = \mathcal{T}$. Hence, $\big( \mathcal{T}_1^\perp \oplus \ldots \oplus \mathcal{T}_K^\perp \big)^\perp \subseteq \mathcal{T}$, and $\mathcal{T}_1^\perp \oplus \ldots \oplus \mathcal{T}_K^\perp \supseteq \mathcal{T}^\perp$.

Finally, by Lemma~\ref{lm:single_dag_model_orthocomp}, the $i$-th subspace $\mathcal{T}_i^\perp = \{\Pi(h_i \mid \mathcal{T}_i^\perp): \forall h \in \mathcal{H}\}$. Then by definition of direct sum, the tangent space $\mathcal{T}$ is the subspace of all functions of the form $\sum_{i = 1}^K \Pi(h_i \mid \mathcal{T}_i^\perp)$ for arbitrary $h_1, \ldots, h_K \in \mathcal{H}$.

\end{proof}

\begin{example}
The Bell scenario model associated with Figure~\ref{fig:ci_models}d has two constraints: $A \Perp B, D$ and $B \Perp A, C$. The operators correspond to these constraints are
\begin{equation}
\begin{aligned}
    \Pi(\cdot \mid \mathcal{T}_1^{\perp}) & = \mathbb{E}[\cdot | a, b, d] - \mathbb{E}[\cdot | a] - \mathbb{E}[\cdot | b, d]
    \\
    \Pi(\cdot \mid \mathcal{T}_2^{\perp}) & = \mathbb{E}[\cdot | b, a, c] - \mathbb{E}[\cdot | b] - \mathbb{E}[\cdot | a, c].
\end{aligned}
\end{equation}
Let $\Gamma(\cdot) = \Pi(\cdot \mid \mathcal{T}_1^{\perp}) + \Pi(\cdot \mid \mathcal{T}_2^{\perp})$.
If $\Gamma$ is a projection operator, then $h - \Gamma(h)$ and $\Gamma(g)$ must be orthogonal, for any $h, g \in \mathcal{H}$. We will compute their inner product. It is not evident that their inner product vanishes.
\begin{equation*}
    \EE \bigg[ (h - \Gamma(h)) \Gamma(g)\bigg] = \EE \bigg[ h \Gamma(g)\bigg] - \EE \bigg[ \Gamma(h) \Gamma(g)\bigg]
\end{equation*}
For the first term
\begin{equation*}
\begin{aligned}
    \mathrm{term 1} & = \EE \bigg[ h \times \bigg( \mathbb{E}[g | A, B, D] - \mathbb{E}[g | A] - \mathbb{E}[g | B, D] + \mathbb{E}[g | B, A, C] - \mathbb{E}[g | B] - \mathbb{E}[g | A, C] \bigg) \bigg]
    \\
    & = \EE \big[ \mathbb{E}[h | A, B, D] \cdot \mathbb{E}[g | A, B, D]\big] + \EE \big[ \mathbb{E}[h | A] \cdot \mathbb{E}[g | A] \big] - \EE \big[ \mathbb{E}[h | B, D] \cdot \mathbb{E}[g | B, D] \big]
    \\
    & + \EE \big[ \mathbb{E}[h | B, A, C] \cdot \mathbb{E}[g | B, A, C] \big] - \EE \big[ \mathbb{E}[h | B] \cdot \mathbb{E}[g | B] \big]- \EE \big[ \mathbb{E}[h | A, C] \cdot \mathbb{E}[g | A, C] \big]
\end{aligned}
\end{equation*}
For the second term
\begin{equation*}
\begin{aligned}
    & \mathrm{term 2} =  \EE \bigg[ \bigg(\mathbb{E}[h | A, B, D] - \mathbb{E}[h | A] - \mathbb{E}[h | B, D] + \mathbb{E}[h | B, A, C] - \mathbb{E}[h | B] - \mathbb{E}[h | A, C] \bigg)
    \\
    & \quad \quad \quad \quad \times \bigg( \mathbb{E}[g | A, B, D] - \mathbb{E}[g | A] - \mathbb{E}[g | B, D] + \mathbb{E}[g | B, A, C] - \mathbb{E}[g | B] - \mathbb{E}[g | A, C] \bigg) \bigg]
    \\
    & = \EE \big[ \mathbb{E}[h | A, B, D] \cdot \mathbb{E}[g | A, B, D]\big] + \EE \big[ \mathbb{E}[h | A] \cdot \mathbb{E}[g | A] \big] - \EE \big[ \mathbb{E}[h | B, D] \cdot \mathbb{E}[g | B, D] \big]
    & & (\mathrm{term 2a})
    \\
    & \quad \quad \quad \quad + \EE \bigg[ \mathbb{E}[h | A, B, D] \bigg( \mathbb{E}[g | B, A, C] - \mathbb{E}[g | B] - \mathbb{E}[g | A, C] \bigg) \bigg]
    & & (\mathrm{term 2b})
    \\
    & - \EE \bigg[ \mathbb{E}[h | A] \bigg( \mathbb{E}[g | A, B, D] - \mathbb{E}[g | A] - \mathbb{E}[g | B, D] \bigg) \bigg]
    - \EE \bigg[ \mathbb{E}[h | A] \bigg( \mathbb{E}[g | B, A, C] - \mathbb{E}[g | B] - \mathbb{E}[g | A, C] \bigg) \bigg]
    & & (\mathrm{term 2c})
    \\
    & - \EE \bigg[ \mathbb{E}[h | B, D] \bigg( \mathbb{E}[g | A, B, D] - \mathbb{E}[g | A] - \mathbb{E}[g | B, D] \bigg) \bigg]
    - \EE \bigg[ \mathbb{E}[h | B, D] \bigg( \mathbb{E}[g | B, A, C] - \mathbb{E}[g | B] - \mathbb{E}[g | A, C] \bigg) \bigg]
    & & (\mathrm{term 2d})
    \\
    & + \EE \big[ \mathbb{E}[h | B, A, C] \cdot \mathbb{E}[g | B, A, C] \big] - \EE \big[ \mathbb{E}[h | B] \cdot \mathbb{E}[g | B] \big]- \EE \big[ \mathbb{E}[h | A, C] \cdot \mathbb{E}[g | A, C] \big]
    & & (\mathrm{term 2e})
    \\
    & \quad \quad \quad \quad + \EE \bigg[ \mathbb{E}[h | B, A, C] \bigg( \mathbb{E}[g | A, B, D] - \mathbb{E}[g | A] - \mathbb{E}[g | B, D] \bigg) \bigg]
    & & (\mathrm{term 2f})
    \\
    & - \EE \bigg[ \mathbb{E}[h | B] \bigg( \mathbb{E}[g | A, B, D] - \mathbb{E}[g | A] - \mathbb{E}[g | B, D] \bigg) \bigg]
    - \EE \bigg[ \mathbb{E}[h | B] \bigg( \mathbb{E}[g | B, A, C] - \mathbb{E}[g | B] - \mathbb{E}[g | A, C] \bigg) \bigg]
    & & (\mathrm{term 2g})
    \\
    & - \EE \bigg[ \mathbb{E}[h | A, C] \bigg( \mathbb{E}[g | A, B, D] - \mathbb{E}[g | A] - \mathbb{E}[g | B, D] \bigg) \bigg]
    - \EE \bigg[ \mathbb{E}[h | A, C] \bigg( \mathbb{E}[g | B, A, C] - \mathbb{E}[g | B] - \mathbb{E}[g | A, C] \bigg) \bigg]
    & & (\mathrm{term 2h})
\end{aligned}
\end{equation*}
Note that $\mathrm{term 1} = \mathrm{term 2a} + \mathrm{term 2e}$. For the others,
\begin{equation*}
\begin{aligned}
    \mathrm{term 2b}
    & = \EE \bigg[ \mathbb{E}[h | A, B, D] \mathbb{E}[g | B, A, C] \bigg]
    - \EE \bigg[ \mathbb{E}[h | B] \mathbb{E}[g | B] \bigg]
    - \EE \bigg[ \mathbb{E}[h | A, B, D] \mathbb{E}[g | A, C] \bigg]
    \\
    \mathrm{term 2c}
    & = - \cancel{\EE \bigg[ \mathbb{E}[h | A] \mathbb{E}[g | A] \bigg]} + \cancel{\EE \bigg[ \mathbb{E}[h | A] \mathbb{E}[g | A] \bigg]} + \EE \bigg[ \mathbb{E}[h | A] \mathbb{E}[g | B, D] \bigg]
    \\
    & \quad - \cancel{\EE \bigg[ \mathbb{E}[h | A] \mathbb{E}[g | A] \bigg]} + \EE \bigg[ \mathbb{E}[h | A] \mathbb{E}[g | B] \bigg] + \cancel{\EE \bigg[ \mathbb{E}[h | A] \mathbb{E}[g | A] \bigg]}
    \\
    \mathrm{term 2d}
    & = - \cancel{\EE \bigg[ \mathbb{E}[h | B, D] \mathbb{E}[g | B, D] \bigg]} + \EE \bigg[ \mathbb{E}[h | B, D] \mathbb{E}[g | A] \bigg] + \cancel{\EE \bigg[ \mathbb{E}[h | B, D] \mathbb{E}[g | B, D] \bigg]}
    \\
    & \quad \quad - \EE \bigg[ \mathbb{E}[h | B, D] \mathbb{E}[g | B, A, C] \bigg] + \EE \bigg[ \mathbb{E}[h | B] \mathbb{E}[g | B] \bigg] + \EE \bigg[ \mathbb{E}[h | B, D] \mathbb{E}[g | A, C] \bigg]
\end{aligned}
\end{equation*}
and
\begin{equation*}
\begin{aligned}
    \mathrm{term 2f}
    & = \EE \bigg[ \mathbb{E}[h | B, A, C] \mathbb{E}[g | A, B, D] \bigg] - \EE \bigg[ \mathbb{E}[h | A] \mathbb{E}[g | A] \bigg] - \EE \bigg[ \mathbb{E}[h | B, A, C] \mathbb{E}[g | B, D] \bigg]
    \\
    \mathrm{term 2g}
    & = - \cancel{\EE \bigg[ \mathbb{E}[h | B] \mathbb{E}[g | B] \bigg]} + \EE \bigg[ \mathbb{E}[h | B] \mathbb{E}[g | A] \bigg] + \cancel{\EE \bigg[ \mathbb{E}[h | B] \mathbb{E}[g | B] \bigg]}
    \\
    & \quad - \cancel{\EE \bigg[ \mathbb{E}[h | B] \mathbb{E}[g | B] \bigg]}
    + \cancel{\EE \bigg[ \mathbb{E}[h | B] \mathbb{E}[g | B] \bigg]}
    + \EE \bigg[ \mathbb{E}[h | B] \mathbb{E}[g | A, C] \bigg]
    \\
    \mathrm{term 2h}
    & = - \EE \bigg[ \mathbb{E}[h | A, C] \mathbb{E}[g | A, B, D] \bigg] + \EE \bigg[ \mathbb{E}[h | A] \mathbb{E}[g | A] \bigg] + \EE \bigg[ \mathbb{E}[h | A, C] \mathbb{E}[g | B, D] \bigg]
    \\
    & \quad \quad - \cancel{\EE \bigg[ \mathbb{E}[h | A, C] \mathbb{E}[g | A, C] \bigg]} + \EE \bigg[ \mathbb{E}[h | A, C] \mathbb{E}[g | B] \bigg] + \cancel{\EE \bigg[ \mathbb{E}[h | A, C] \mathbb{E}[g | A, C] \bigg]}
\end{aligned}
\end{equation*}
Then
\begin{equation*}
\begin{aligned}
\mathrm{term 2}
& = \mathrm{term 1}
\\
& + \EE \bigg[ \mathbb{E}[h | A, B, D] \mathbb{E}[g | B, A, C] \bigg] - \cancel{\EE \bigg[ \mathbb{E}[h | B] \mathbb{E}[g | B] \bigg]} - \EE \bigg[ \mathbb{E}[h | A, B, D] \mathbb{E}[g | A, C] \bigg]
\\
& + \EE \bigg[ \mathbb{E}[h | A] \mathbb{E}[g | B, D] \bigg] + \EE \bigg[ \mathbb{E}[h | A] \mathbb{E}[g | B] \bigg]
\\
& + \EE \bigg[ \mathbb{E}[h | B, D] \mathbb{E}[g | A] \bigg] - \EE \bigg[ \mathbb{E}[h | B, D] \mathbb{E}[g | B, A, C] \bigg] + \cancel{\EE \bigg[ \mathbb{E}[h | B] \mathbb{E}[g | B] \bigg]} + \EE \bigg[ \mathbb{E}[h | B, D] \mathbb{E}[g | A, C] \bigg]
\\
& + \EE \bigg[ \mathbb{E}[h | B, A, C] \mathbb{E}[g | A, B, D] \bigg] - \cancel{\EE \bigg[ \mathbb{E}[h | A] \mathbb{E}[g | A] \bigg]} - \EE \bigg[ \mathbb{E}[h | B, A, C] \mathbb{E}[g | B, D] \bigg]
\\
& + \EE \bigg[ \mathbb{E}[h | B] \mathbb{E}[g | A] \bigg] + \EE \bigg[ \mathbb{E}[h | B] \mathbb{E}[g | A, C] \bigg]
\\
& - \EE \bigg[ \mathbb{E}[h | A, C] \mathbb{E}[g | A, B, D] \bigg] + \cancel{\EE \bigg[ \mathbb{E}[h | A] \mathbb{E}[g | A] \bigg]} + \EE \bigg[ \mathbb{E}[h | A, C] \mathbb{E}[g | B, D] \bigg] + \EE \bigg[ \mathbb{E}[h | A, C] \mathbb{E}[g | B] \bigg]
\end{aligned}
\end{equation*}
It is not evident that $\mathrm{term 2} = \mathrm{term 1}$.
\end{example}

\section{Proofs in Section \ref{sec:application}}

\begin{thma}{\ref{thm:efficiency_improvement}}
    Consider the Markov model with $K \geq 1$ constraint in Theorem~\ref{thm:general_ci_model_orthocomp}, and let $\varphi_0$ be an influence function of the target parameter $\psi(p)$. For any sequence of integers $i_1, i_2, \ldots, i_{M}$ where each $i_m \in \{1, \ldots, K\}$ for all $m \in \{1, \ldots, M\}$, define
    \begin{equation}
    \begin{aligned}
        \varphi_m := \varphi_{m-1} - \Pi(\varphi_{m-1} \mid \mathcal{T}_{i_m}^\perp), \quad \forall m \in \{1, \ldots, M\}.
    \end{aligned}
    \end{equation}
    In words, $\varphi_m$ is the projection of $\varphi_{m-1}$ onto the tangent space $\mathcal{T}_{i_m}$ of the DAG model corresponding to the $i_m$-th constraint. Then $\varphi_0, \ldots, \varphi_{M}$ is a sequence of influence functions of $\psi(p)$ with non-increasing variance, meaning $E[\varphi_{m}^2] \leq E[\varphi_{m-1}^2]$ for all $m \in \{1, \ldots, M\}$.
\end{thma}
\begin{proof}
    The proof of this theorem follows from the discussions in Section~\ref{sec:application:improve} of the main paper and Section~\ref{sec:additional_semi_param} of the Appendix.
\end{proof}

The iterative process described in Theorem~\ref{thm:efficiency_improvement} along with the explicit form of the projections $\Pi(\cdot \mid \mathcal{T}_{i_m}^\perp)$ in Theorem~\ref{thm:general_ci_model_orthocomp} suggests a natural procedure for constructing estimators. For example, given an IF of the following form (obtained from one step of the iteration):
\begin{equation}
\begin{aligned}
    \varphi_{m}
    & := \varphi_{m-1} - \Pi(\varphi_{m-1} \mid \mathcal{T}_{i_{m}}^\perp)
    \\
    & = \varphi_{m-1} - E[\varphi_{m-1} \mid \mathbf{x}_{i_{m}}, \mathbf{y}_{i_{m}}, \mathbf{z}_{i_{m}}] + E[\varphi_{m-1} \mid \mathbf{x}_{i_{m}}, \mathbf{z}_{i_{m}}] + E[\varphi_{m-1} \mid \mathbf{y}_{i_{m}}, \mathbf{z}_{i_{m}}] - E[\varphi_{m-1} \mid \mathbf{z}_{i_{m}}],
\end{aligned}
\end{equation}
an estimator may be constructed by adopting additional nuisance models for the four additional expectation terms.
For some specific IFs, we maybe able to simplify $\varphi_{m}$ because some terms from $\varphi_{m-1}$ and the expectations may cancel.

Note that this is only a single step in the iterative process, where multiple operators are applied sequentially. If we want to use an $M$-th improvement of some initial IF $\varphi_0$, the expression will involve nested expectations, and may become challenging to evaluate in practice. However, in some cases, we may derive explicit estimators that demonstrate efficiency improvement, as we demonstrate via the following example.

\begin{example}

Consider the Bell scenario model corresponding to Figure~\ref{fig:ci_models}d, whose constraints are $A \Perp B,D$ and $B \Perp A,C$. Suppose that we want to estimate the target parameter $\psi(p) = \EE[C \mid a_0]$ instead of $\EE[D \mid a_0]$, as $D$ and $A$ are independent in this model. The influence function for this parameter in the saturated model is:
\begin{equation}
    \varphi_0(a, c) = \frac{\mathbb{I}(a = a_0)}{p(a)} (c - E[C \mid a])
\end{equation}
Let $\mathscr{P}_1, \mathscr{P}_2$ be the single independence DAG models corresponding to the constraints $A \perp\!\!\!\perp B,D$ and $B \perp\!\!\!\perp A,C$, respectively. According to Theorem~\ref{thm:general_ci_model_orthocomp}
\begin{align*}
    \Pi(\cdot \mid \mathcal{T}_1^{\perp}) & = \mathbb{E}[\cdot | a, b, d] - \mathbb{E}[\cdot | a] - \mathbb{E}[\cdot | b, d]
    \\
    \Pi(\cdot \mid \mathcal{T}_2^{\perp}) & = \mathbb{E}[\cdot | b, a, c] - \mathbb{E}[\cdot | b] - \mathbb{E}[\cdot | a, c].
\end{align*}

We will derive the first and second improvements of $\varphi_0$, namely $\varphi_1$ and $\varphi_2$, respectively, using the iterative method in Theorem~\ref{thm:efficiency_improvement}, and explicitly derive the corresponding estimators.

\textbf{First improvement:} 
\begin{equation}
\begin{aligned}
    \varphi_1 = \varphi_0 - \mathbb{E}[\varphi_0 | a, b, d] + \mathbb{E}[\varphi_0 | a] + \mathbb{E}[\varphi_0 | b, d],
\end{aligned}
\end{equation}
where
\begin{equation}
\begin{aligned}
\EE[\varphi_0 \mid a, b, d] & = \frac{\mathbb{I}(a = a_0)}{p(a)} (\EE[C \mid a, b, d] - \EE[C \mid a])
\\
\EE[\varphi_0 \mid a] & = \frac{\mathbb{I}(a = a_0)}{p(a)} (\EE[C \mid a] - \EE[C \mid a]) = 0
\\
\EE[\varphi_0 \mid b, d] & = \int \EE[\varphi_0 \mid a', b, d] p(a' \mid b, d) da'
\\
& = \int \EE[\varphi_0 \mid a', b, d] p(a') da' & (A \perp\!\!\!\perp B, D)
\\
& = \int \bigg( \frac{\mathbb{I}(a' = a_0)}{p(a')} (\EE[C \mid a', b, d] - \EE[C \mid a']) \bigg) p(a') da'
\\
& = \EE[C \mid a_0, b, d] - \EE[C \mid a_0]
\end{aligned}
\end{equation}
Substituting these expressions, cancelling terms, and noting that $\psi(p) = \EE[C \mid a_0]$, we get
\begin{equation}
\begin{aligned}
    \varphi_1(a,b,c,d)
    & = \frac{\mathbb{I}(a = a_0)}{p(a)} (c - \EE[C \mid a, b, d])
    + \EE[C \mid a_0, b, d] - \psi(p)
\end{aligned}
\end{equation}
Let $\hat{p}(a)$ be an estimate of the distribution $p(a)$, and $\hat{\EE}[C \mid a, b, d]$ be an estimate of the conditional mean. The estimation equation method gives us an estimate $\hat{\psi}_1$ of the target $\psi(p)$ via
\begin{equation}
\begin{aligned}
    \hat{\psi}_1 = \frac{1}{n} \sum_{i=1}^n \left( \frac{\mathbb{I}(A_i = a_0)}{\hat{p}(A_i)} (C_i - \hat{\EE}[C \mid A_i, B_i, D_i]) + \hat{\EE}[C \mid a_0, B_i, D_i] \right).
\end{aligned}
\end{equation}
Note that this estimator is the augmented inverse probability weighted (AIPW) estimator arising in causal inference and missing data problems. For example, $A$ is the treatment, $C$ is the outcome, and $\{B, D\}$ is the observed confounder. Then the Bell scenario is the observed confounder scenario, with randomization on $A$, explaining this ``coincidence''.

\textbf{Second improvement:} 
\begin{equation}
\begin{aligned}
    \varphi_2 = \varphi_1 - E[\varphi_1 \mid b, a, c] + E[\varphi_1 \mid b] + E[\varphi_1 \mid a, c],
\end{aligned}
\end{equation}
where
\begin{equation}
\begin{aligned}
\EE[\varphi_1 \mid b]
& = \int \frac{\mathbb{I}(a' = a_0)}{p(a')} c' \cdot p(a', c' \mid b) da' dc' - \int \frac{\mathbb{I}(a' = a_0)}{p(a')} \bigg( \int \EE[C \mid a', b, d'] p(d' \mid b, a') d d' \bigg) p(a' \mid b) da'
\\
& + \int \EE[C \mid a_0, b, d'] p(d' \mid b) dd' - \psi(p)
\\
& = \EE[C \mid a_0] - \int \EE[C \mid a_0, b, d'] p(d' \mid b, a_0) dd'
\quad \quad \quad (A, C \perp\!\!\!\perp B)
\\
& + \int \EE[C \mid a_0, b, d'] p(d' \mid b) dd' - \psi(p)
\\
& = \int \EE[C \mid a_0, b, d'] p(d' \mid b) dd' - \int \EE[C \mid a_0, b, d'] p(d' \mid b, a_0) dd'
\quad \quad \quad (\EE[C \mid a_0] = \psi(p))
\\
& = \int \EE[C \mid a_0, b, d'] p(d' \mid b, a_0) dd' - \int \EE[C \mid a_0, b, d'] p(d' \mid b, a_0) dd'
\quad \quad \quad (D \Perp A \mid B)
\\
& = 0.
\end{aligned}
\end{equation}
\begin{equation}
\begin{aligned}
\EE[\varphi_1 \mid b, a, c] & = \frac{\mathbb{I}(a = a_0)}{p(a)} \left(c - \int \EE[C \mid a, b, d'] p(d' \mid b, a, c) d d' \right)
\\
& + \int \EE[C \mid a_0, b, d'] p(d' \mid b, a, c) dd' - \psi(p)
\\
\EE[\varphi_1 \mid a, c] & = \frac{\mathbb{I}(a = a_0)}{p(a)} \left(c - \int \EE[C \mid a, b', d'] p(d' \mid b', a, c) p(b') dd' db' \right)
\quad \quad \quad (B \perp\!\!\!\perp A, C)
\\
& + \int \EE[C \mid a_0, b', d'] p(d' \mid b', a, c) p(b') dd' db' - \psi(p)
\quad \quad \quad (B \perp\!\!\!\perp A, C).
\end{aligned}
\end{equation}
Then
\begin{equation}
\begin{aligned}
\varphi_2(a,b,c,d)
& = \frac{\mathbb{I}(a = a_0)}{p(a)} (c - \EE[C \mid a, b, d]) + \EE[C \mid a_0, b, d] - \psi(p)
\\
& + \frac{\mathbb{I}(a = a_0)}{p(a)} \left(\int \EE[C \mid a, b, d'] p(d' \mid b, a, c) dd' - \int \EE[C \mid a, b, d'] p(d' \mid b', a, c) p(b') dd' db' \right)
\\
&  - \int \EE[C \mid a_0, b, d'] p(d' \mid b, a, c) dd' + \int \EE[C \mid a_0, b, d'] p(d' \mid b', a, c) p(b') dd' db'
\end{aligned}
\end{equation}
We posit models for $\hat{p}(a), \hat{\EE}[C \mid a, b, d], \hat{p}(d \mid a, c, b)$ and $\hat{p}(b)$. Importantly, these models must be consistent, i.e., there is a valid set of probability distributions $p(a,b,c,d)$ satisfying the constraints and yield these models. Enforcing consistency is not always obvious. If the models are consistent, then the estimation equation method gives us the estimator $\hat{\psi}_2$ for the target $\psi(p)$
\begin{equation}
\begin{aligned}
\hat{\psi}_2
& = \frac{1}{n} \sum_{i=1}^n \left( \frac{\mathbb{I}(A_i = a_0)}{\hat{p}(A_i)} (C_i - \hat{\EE}[C \mid A_i, B_i, D_i]) + \hat{\EE}[C \mid a_0, B_i, D_i] \right)
\\
& + \frac{1}{n} \sum_{i=1}^n \frac{\mathbb{I}(A_i = a)}{\hat{p}(A_i)} \left( \int \hat{\EE}[C \mid A_i, B_i, d'] \hat{p}(d' \mid B_i, A_i, C_i) dd' - \int \hat{\EE}[C \mid A_i, b', d'] \hat{p}(d' \mid b', A_i, C_i) \hat{p}(b') dd' db' \right)
\\
& + \frac{1}{n} \sum_{i=1}^n \left( - \int \hat{\EE}[C \mid a_0, B_i, d'] \hat{p}(d' \mid B_i, A_i, C_i) dd' + \int \hat{\EE}[C \mid a_0, b', d'] \hat{p}(d' \mid b', A_i, C_i) \hat{p}(b') dd' db' \right).
\end{aligned}
\end{equation}
\end{example}

\end{document}